\documentclass[11pt,letterpaper]{article}
\usepackage[letterpaper, margin=0.75in]{geometry}
\usepackage{graphicx,color}
\usepackage[cp850]{inputenc}
\usepackage[T1]{fontenc}
\usepackage[english]{babel}
\usepackage{rotating}
\usepackage[f]{esvect}
\usepackage{amsmath, amsthm, amssymb}
\usepackage{rotating}
\usepackage{lscape}
\usepackage{mathrsfs}
\usepackage{mathtools}
\usepackage[toc,page]{appendix}
\usepackage{hyphenat}

\usepackage[normalem]{ulem}

\def\tablenotes{\bgroup\parfillskip=0pt plus 1fil
\leftskip=0pt\relax \rightskip=0pt
\vskip2pt\footnotesize}
\def\endtablenotes{\vskip1pt\egroup}

\newtheorem{theorem}{Theorem}[section]
\newtheorem{proposition}[theorem]{Proposition}

\newtheorem{lemma}[theorem]{Lemma}
\newtheorem{remark}[theorem]{Remark}

\renewcommand{\epsilon}{\varepsilon}

\newcommand{\TimeDeriv}{\frac{\textrm{d}}{\textrm{dt}}}

\renewcommand{\epsilon}{\varepsilon}
\renewcommand{\leq}{\leqslant}
\renewcommand{\geq}{\geqslant}
\renewcommand{\d}{\mathrm{d}}

\newcommand{\eq}[1]{\eqref{#1}}

\renewcommand{\phi}{\varphi}
\renewcommand{\epsilon}{\varepsilon}
 
\allowdisplaybreaks

\usepackage[square,authoryear,sort]{natbib}

\usepackage{graphicx}
\usepackage{overpic}
\usepackage{tikz}
\numberwithin{equation}{section}

\makeatletter
\renewcommand{\@biblabel}[1]{#1\hfill \hspace{-0.2cm}}
\makeatother

\begin{document}
\title{A Recipe for State Dependent Distributed Delay Differential Equations }
\author{%
  Tyler Cassidy\affil{1},
  Morgan Craig\affil{2,3}
  and
  Antony R. Humphries\affil{1,3}
}


\author{  { \sc Tyler Cassidy$^1$,}
 { \sc Morgan Craig$^{2,3}$}
  and
  {\sc Antony R. Humphries$^{1,3}$} \\[2pt]
$^1$ Department of Mathematics and Statistics, McGill University, \\
 805 Sherbrooke Street West, Montreal, H3A 0B9, Canada.\\[6pt]
 $^2$ D\'{e}partement de math\'{e}matiques et de statistique, Universit\'{e} de Montr\'{e}al, \\
  2920 chemin de la Tour, Montr\'{e}al, H3T 1J4, Canada \\[6pt]
 $^3$ Department of Physiology, McGill University, \\
  3655 Promenade Sir-William-Osler, Montreal, H3G 1Y6, Canada.\\[6pt]
}
\maketitle

\begin{abstract}
{We use the McKendrick equation with variable ageing rate and randomly distributed maturation time to derive a state dependent distributed delay differential equation. We show that the resulting delay differential equation preserves non-negativity of initial conditions and we characterise local stability of equilibria. By specifying the distribution of maturation age, we recover state dependent discrete, uniform and gamma distributed delay differential equations. We show how to reduce the uniform case to a system of state dependent discrete delay equations and the gamma distributed case to a system of ordinary differential equations. To illustrate the benefits of these reductions, we convert previously published transit compartment models into equivalent distributed delay differential equations. }
\end{abstract}

\section{Introduction}\label{Sec:Introduction}
Age structured population models have been used extensively in mathematical biology throughout the past 90 years \citep{Mckendrick1925,Trucco1965} (see \citep{Metz1986} for a review). These age structured models describe the progression of individuals through an ageing process by using partial differential equations (PDEs), that can, in certain cases, be reduced to a delay differential equation (DDE)  \citep{Metz1986,Smith1993,Craig2016}.  When individuals exit the ageing process in a deterministic manner upon reaching a threshold maturation age, the age structured model is typically reduced to a discrete DDE.

In many populations, the speed at which an individual matures is often only weakly coupled to chronological time and is dynamically controlled by the availability of resources. Consequently, when considering the age of an individual in a population, it is the biological age -- and not the chronological age-- that is of interest. It is possible to allow for this dynamic accumulation of biological age by including a variable ageing rate in an age structured PDE model. PDE models with variable ageing rates and threshold maturation rates can be reduced to state dependent discrete DDEs. State dependent delays considerably complicate the study of these models, but incorporate external control of the maturation process and increase physiological relevance.

However, imposing a threshold maturation age does not account for population heterogeneity and implicitly assumes a homogeneous maturation age. Given the importance of individual differences in a population, it is important that intraspecies heterogeneity is included in mathematical models. In light of these observations, we develop a technique to explicitly incorporate maturation age heterogeneity and external control of age accumulation by providing a framework for state dependent distributed DDEs. State dependent distributed DDEs account for a measure of population heterogeneity not present in discrete DDE models while retaining external control of the ageing process. Therefore, distributed DDEs offer a physiologically more realistic manner to model ageing processes in populations \citep{Cassidy2018}.

To derive a state dependent distributed DDE, we consider a general age structured model with a variable ageing rate. We eschew a deterministic maturation process (which would lead to state dependent discrete DDEs), and instead utilise a randomly distributed maturation age $A$.
This random variable defines a density function $K_A(t)$ through
\begin{equation}
K_A(t) = \lim \limits_{\Delta t \to 0} \frac{\mathbb{P}\left[t\leq A \leq t+\Delta t \right]}{\Delta t},
\label{Eq:DensityDefinition}
\end{equation}
which satisfies
\begin{equation*}
\int_0^{\infty} K_A(t) \d t = 1 \quad \textrm{and} \quad K_A(t) \geq 0 \quad \forall t \geq 0.
\end{equation*}
As shown by \citet{Craig2016,Otto2017} and \citet{Bernard2016}, replacing existing discrete delays with state dependent delays requires careful attention to how solutions pass across the maturation boundary. \citet{Craig2016} derived a ``correction'' factor to ensure that individuals are not spuriously created or destroyed during maturation. Our work generalises the correction factor derived by \citet{Craig2016} for state dependent discrete DDEs to any state dependent DDE. Specifically, our derivation does not rely on a smoothness argument, but arises naturally from the age structured PDE after a careful derivation of the maturation rate.

We show how the age structured PDE can be reduced to a state dependent distributed DDE. For specific densities $K_A(t)$, we show equivalence between the state dependent distributed DDE and state-dependent discrete DDEs with one or two delays or a finite dimensional systems of ordinary differential equations (ODEs). These equivalences arise from the explicit consideration of the ageing process modelled by the distributed DDEs. By applying the linear chain technique to the age variable, instead of the time variable, we are able to establish the desired equivalences. As there is not an available all purpose numerical method capable of solving distributed DDEs, these equivalences allow for the model to be analysed as a DDE and simulated using the highly efficient established techniques for discrete DDEs or ODEs.  To illustrate the benefits of the techniques developed here, we consider two previously published models of hematopoietic cell production and show how using distributed DDEs can simplify the analysis of the resulting model.

The structure of the article is as follows. In Section~\ref{Sec:AgeStructuredPDEReduction}, we study the McKendrick equation for a generic population with a variable ageing rate and random maturation time. By solving the PDE using the method of characteristics, we derive a state-dependent distributed DDE for the general density $K_A(t)$ in Theorem~\ref{Theorem:GeneralDDE}. We discuss the naturally arising ``correction'' factor in Section~\ref{Sec:CorrectionFactor}. To illustrate the benefits of reducing age structured models to DDEs, we show that the resulting DDE preserves non-negativity of initial conditions and perform stability analysis to study the local stability of equilibria in Section~\ref{Sec:AnalysisofDDE}. By specifying $K_A(t)$ to be the degenerate distribution, we recover a state-dependent discrete DDE in Section~\ref{Sec:StateDependentDDE}. Next, we consider uniform distributions and the equivalent two delay DDE in Section~\ref{Sec:CompactSupportDistribution}. In Section~\ref{Sec:GammaDistributedDDE}, we study a gamma distributed DDE. Through a generalization of the linear chain technique to include a variable transit rate, we show how this gamma distributed DDE can be reduced to a finite dimensional system of transit compartment ODEs in Section~\ref{Sec:FiniteDimensionalRepresentation}.  In Section~\ref{Sec:Examples}, we formalize the link between variable transit rate compartment models and state dependent delayed processes by converting two previously published transit compartment models to the corresponding distributed DDEs. Finally, we summarize our results with a brief conclusion.

\section{From McKendrick Type Equations to State Dependent Delays}\label{Sec:AgeStructuredPDEReduction}

Consider a population divided into immature and  mature compartments in which only mature individuals reproduce. Let $n(t,a)$ denote the number of immature individuals  at time $t$ with age $a$ and $x(t)$ denote the number of mature members of the population at time $t$. The purpose of this section is to establish a state dependent distributed DDE model for $x(t)$.

We begin with an age structured PDE for the immature population, $n(t,a)$. Immature individuals progress through maturation with a variable ageing rate $V_a(t)$, where $V_a(t)$ satisfies
\begin{equation*}
0< V_a^{min} \leq V_a(t) \leq V_a^{max} < \infty.
\end{equation*}
Following \citet{Mckendrick1925}, the PDE describing $n(t,a)$ is
\begin{equation}
\left.
\begin{aligned}
\partial_t n(t,a) + V_a(t)\partial_a n(t,a) & = -\left[\mu(x(t))+h(a)\right]n(t,a)\\
V_a(t) n(t,0)  =  \beta x(t) \quad t \geq t_0;  & \quad n(t_0,a) = f(a) \geq 0 \quad \forall a \in (0, \infty ).
\end{aligned}
\right \}
\label{Eq:McKendrickAgePDE}
\end{equation}
The boundary condition $V_a(t) n(t,0) = \beta x(t)$ that we impose links the creation of immature individuals $n(t,0)$ with the birth rate $\beta x(t)$. The presence of $V_a(t)$ in this boundary term can be understood from the conveyor belt analogy \citep{Mahaffy1998,Bernard2016}.
In the following, we assume $\beta >0$. 
The initial conditions $n(t_0,a) = f(a) \geq 0,$ describes immature individuals with non-zero age at time $t_0$.

The death rate of immature individuals is given by $\mu(x(t))$ while transition from the immature state to the mature state is modelled by $ h(a)$. It is important to note that the transition rate is a function of the age of individuals at time $t$. Since we expect a link between time and physiological age, we will write $a(t)$. Later, we  formalize the weakly coupled relationship between biological and chronological age and justify this notation by finding the characteristics of \eqref{Eq:McKendrickAgePDE}.

We begin by deriving the transition rate from immaturity to maturity, $h(a(t))$. As mentioned, we assume that the age at which an individual matures is a non-negative random variable $A$ with density function $K_A(t)$. The transition rate, $h(a(t))$, is the instantaneous change in probability that an individual matures at age $a(t+\Delta t)$, given that the individual has not matured at age $a(t)$. Formally, using the definition of conditional probability,
\begin{equation*}
h(a(t)) = \lim \limits_{\Delta t \to 0} \frac{\mathbb{P}\left[a(t)\leq A \leq a(t+\Delta t)| A \geq a(t) \right]}{\Delta t} = \lim \limits_{\Delta t \to 0} \frac{\mathbb{P}\left[a(t)\leq A \leq a(t+\Delta t) \right]}{ \mathbb{P}[A \geq a(t)] \Delta t}.
\end{equation*}
Multiplying by unity gives
\begin{equation*}
h(a(t)) = \frac{1}{  \mathbb{P}[A \geq a(t)] }\lim \limits_{\Delta t \to 0} \frac{\mathbb{P}\left[a(t)\leq A \leq a(t+\Delta t) \right]}{ (a(t+\Delta t)-a(t)) } \frac{a(t+\Delta t)-a(t)}{\Delta t}.
\end{equation*}
By \eqref{Eq:DensityDefinition} and the derivative of $a(t)$, we obtain
\begin{equation}
h(a(t)) = \frac{K_A(a(t))}{1-\int_{0}^{a(t)} K_A(\sigma)\d \sigma} \TimeDeriv a(t).
\label{Eq:HazardRateDefinition}
\end{equation}
The transition (or maturation) rate, $h(a(t))$, is known as the hazard rate of the random variable $A$ and has applications in modelling failure rates \citep{Cox1972,Kaplan1958}. \citet{Metz1986} derived the identical expression for $h(a(t))$ without considering the conditional maturation probability.

It is possible that immature individuals create multiple mature individuals upon transitioning to the mature compartment (i.e mitosis), so we model the influx rate into the mature compartment as a function
\begin{equation*}
F\left( x(t),\int_{0}^{\infty} h(s)n(t,s)\d s \right),
\end{equation*}
where the integral term
\begin{equation}
\int_{0}^{\infty} h(s)n(t,s)\d s
\label{Eq:NumberOfIndividualsMaturing}
\end{equation}
is the number of immature individuals that reach maturity at time $t$. If mature individuals are cleared at a population dependent rate $\gamma(x(t))$, then the mature population satisfies
\begin{equation}
\left.
\begin{aligned}
\TimeDeriv x(t) & = F\left( x(t),\int_{0}^{\infty} h(s)n(t,s)\d s \right) - \gamma(x(t)) x(t) \\
x(0) & = x_0. \\
\end{aligned}
\right \}
\label{Eq:MaturePopulationDE}
\end{equation}

We are now able to establish equivalence between the system of equations describing the populations $x(t)$ and $n(t,a)$ and a distributed DDE. To do this, we partially solve the PDE~\eqref{Eq:McKendrickAgePDE} using the method of characteristics.

\begin{theorem}[State-Dependent Distributed DDE]\label{Theorem:GeneralDDE}
Let the immature population $n(t,a)$ satisfy the McKendrick age structured PDE~\eqref{Eq:McKendrickAgePDE} with the distribution dependent transition rate $h(a(t))$ \eqref{Eq:HazardRateDefinition}. Assume that the mature population $x(t)$ is given by \eqref{Eq:MaturePopulationDE}.

Then, the mature population $x(t)$ satisfies the initial value problem (IVP)
\begin{equation}
\begin{aligned}
\TimeDeriv x(t) & =  F\left( x(t),  \int_{0}^{\infty}  \beta x(t-\phi) \frac{V_a(t)}{V_a(t-\phi)}\exp \left[ - \int_{t-\phi}^t \mu(x(s)) \d s\right] K_A\left( \int_{t-\phi}^t V_a(s)\d s \right) \d \phi  \right) \\
& \qquad {} - \gamma(x(t)) x(t)
\end{aligned}
\label{Eq:MatureDDE}
\end{equation}
with initial data
\begin{equation*}
x(s) = \rho(s)  \quad \forall s \in (-\infty,t_0].
\end{equation*}
\end{theorem}

\begin{proof}
The characteristics of equation~\eqref{Eq:McKendrickAgePDE} satisfy
\begin{equation}
\frac{\d }{\d \phi} t(\phi) = 1, \quad \textrm{and} \quad  \frac{\d }{\d t}  a(t) = V_a(t),
\label{Eq:CharacteristicLinesEq}
\end{equation}
and hence are given by
\begin{equation*}
t = \phi+ T_0 \quad \textrm{and} \quad a(t) = \int_{T_0}^{t} V_a(x)\d x+a_0.
\end{equation*}
Along the characteristics, the age structured PDE~\eqref{Eq:McKendrickAgePDE} becomes
\begin{equation}
\frac{\d}{\d t} n(t,a(t)) = -\left[ \mu(x(t)+ \frac{K_A(a(t))}{1-\int_{0}^{a(t)} K_A(\sigma)\d \sigma} V_a(t) \right] n(t,a(t)).
\label{Eq:ODEOnCharacteristic}
\end{equation}
Equation~\eqref{Eq:ODEOnCharacteristic} is a separable differential equation with solution
\begin{equation*}
n(t,a(t)) = n(T_0,a_0)\exp \left[ - \int_{T_0}^t \mu(x(s)) \d s\right] \left(1-\int_{0}^{a(t)} K_A(\sigma)\d \sigma \right).
\end{equation*}
If $a_0 = 0$, we use the boundary condition of \eqref{Eq:McKendrickAgePDE} to find
\begin{equation}
n(t,a(t)) = \frac{\beta x(T_0)}{V_a(T_0)} \exp \left[ - \int_{T_0}^t \mu(x(s)) \d s\right] \left(1-\int_{0}^{a(t)} K_A(\sigma)\d \sigma \right),
\label{Eq:TotalCellsAgeingProcess}
\end{equation}
while, if $a_0 >0$, the initial condition of \eqref{Eq:McKendrickAgePDE} gives
\begin{equation*}
n(t,a(t)) = f(a_0)\exp \left[ - \int_{t_0}^{t} \mu(x(s)) \d s\right] \left(1-\int_{0}^{a(t)} K_A(\sigma)\d \sigma \right).
\end{equation*}
To establish an equivalence between the PDE \eqref{Eq:McKendrickAgePDE}
and the distributed DDE \eq{Eq:MatureDDE}, it is necessary to define suitable initial data
$x(s) = \rho(s)$ for $s<t_0$ for the DDE. To do this, it is natural to assume that an an immature individual with positive age $a >0$ at time $t_0$ was  born at sometime $s<t_0$. Since the PDE \eqref{Eq:McKendrickAgePDE} is not defined for $s < t_0$, we are free to prescribe fixed values for $V_a(s) = V_a^*$ and $\mu(x(s)) = \mu^*$ for $ s < t_0$. Then, imposing that individuals born at time $s < t_0$ evolved according to the McKendrick Equation, we have $a = V_a^*(t_0-s)$, or $s = t_0-a/V_a^*$. Hence, the initial condition $f(a)$ defines the history function $\rho$ through
\begin{equation}
f(a) = \frac{\beta}{V_a^*} \rho (t_0-a/V_a^*)\exp\left[ \int_{t_0-a/V_a^*}^{t_0} - \mu^*\d s\right].
\label{Eq:HistoryFunctionDefinition}
\end{equation}
Therefore defining $x(s) = \rho(s)$, for $s < t_0$, this way, the solution \eqref{Eq:TotalCellsAgeingProcess} applies. 

Now, we finalize the link between the age structured PDE and the distributed DDE by following the characteristic curves until they intersect with the $a=0$ axis. Along the characteristic curves, at time $t$, individuals born at time $T_0= t-\phi $ have age
\begin{equation*}
a_{t}(\phi) = \int_{T_0}^t V_a(x)\d x = \int_{t-\phi}^t V_a(x)\d x
\end{equation*} 
for $\phi >0.$ So we have
\begin{equation*}
n(t,a_{t}(\phi) ) = \frac{\beta x(t-\phi)}{V_a(t-\phi)} \exp \left[ - \int_{t-\phi}^t \mu(x(s)) \d s\right] \left(1-\int_{0}^{a_{t}(\phi) } K_A(\sigma)\d \sigma \right)
\end{equation*}
At time $t$, the rate at which individuals mature  is
\begin{align}\label{Eq:MaturationInput}
\hspace{-.4cm} \int_{0}^{\infty} h(a_{t}(\phi) )n(t,a_{t}(\phi) )\d \phi &  =   \int_{0}^{\infty}  K_A(a_{t}(\phi) ) \beta x(t-\phi) \frac{V_a(t)}{V_a(t-\phi)} \exp \left[ - \int_{t-\phi}^t \mu(x(s)) \d s\right] \d \phi.
\end{align}
By defining, for any density $K_A(t)$,
\begin{equation}
A_K(x(t)) := \int_{0}^{\infty} K_A(a(\phi)) \frac{ \beta x(t-\phi)  }{V_a(t-\phi)}  \exp \left[ - \int_{t-\phi}^t \mu(x(s)) \d s\right] \d \phi,
\label{Eq:AKDefinition}
\end{equation}
we have
\begin{equation*}
\hspace{-.2cm} \int_{0}^{\infty} h\left( a_{\phi}(t) \right) n\left( t,a_{\phi}(t)\right) \d \phi  =  V_a(t) A_K(x(t)).
\end{equation*}
Consequently, using \eqref{Eq:AKDefinition} and defining the history $\rho(s)$ according to \eqref{Eq:HistoryFunctionDefinition}, we have established the equivalence between the system of \eqref{Eq:McKendrickAgePDE} and \eqref{Eq:MaturePopulationDE} with the distributed DDE \eqref{Eq:MatureDDE}.
\end{proof}

\subsection{Accounting for the Random Maturation Threshold}\label{Sec:CorrectionFactor}

Further inspection of equation~\eqref{Eq:MaturationInput} reveals a ratio of ageing speeds $V_a(t)/V_a(t-\phi)$ in the integral term
\begin{equation*}
 \int_{0}^{\infty} h(a(\phi))n(t,a(\phi))\d \phi = \int_{0}^{\infty}  \beta x(t-\phi)  \frac{V_a(t)}{V_a(t-\phi)}\exp \left[ - \int_{t-\phi}^t \mu(x(s)) \d s\right] K_A(\sigma)\d \phi.
\end{equation*}
The ratio of ageing velocities at the entrance and exit of the ageing process acts as a ``correction factor''.  As shown by \citet{Bernard2016} and \citet{Craig2016}, models without the correction factor allow for spurious creation of individuals during maturation and some state-dependent DDEs have missed this important correction factor. Solutions of models without this correction factor do not necessarily preserve nonnegativity of initial data \citep{Bernard2016}.

\citet{Craig2016} derived the correction factor by carefully accounting for the number of cells crossing the maturation threshold in a discrete state-dependent DDE. In discrete DDEs, individuals mature following a deterministic process after accruing a specific threshold age, so the maturation boundary is well-defined. The derivation of the correction factor was based on the smoothness of the solution crossing the fixed maturation boundary. However, the idea of a fixed maturation boundary does not extend to random maturation ages. Consequently, the derivation of the correction factor by \citet{Craig2016} does not generalise to generic distributed DDEs.

Our derivation of the state-dependent distributed DDE produces the same correction factor through the instantaneous maturation probability, $h(a(t))$. The derivation of $h(a(t))$ in equation~\eqref{Eq:HazardRateDefinition} produces the term $V_a(t)$ by accounting for the change of maturation probability due to the variable accumulation of age at time $t$. For the degenerate distribution, as shown in Section~\ref{Sec:StateDependentDDE}, we obtain precisely the same ratio as \citet{Craig2016}.


\section{Properties of State Dependent Delay Differential Equations}\label{Sec:AnalysisofDDE}

Replacing an age structured PDE by a DDE eliminates the need to explicitly model the ageing populations, which can be difficult to measure experimentally. DDEs offer a natural framework that explicitly incorporates delays and identifies the relationship between the current and past states. This can facilitate communication between mathematical biologists and biologists and physiologists. In particular, the explicit presence of the delay term allows for simple calculation of mean delay time that is important for translatability. As shown by \citet{deSouza2017}, models of delayed processes without DDEs do not always accurately calculate the mean delay time.  However, DDEs typically define infinite dimensional semi-dynamical systems, which can introduce mathematical difficulties.

As we have seen in Theorem~\ref{Theorem:GeneralDDE}, partially solving an age structured PDE may lead to a DDE. As such, analysing these partially solved systems can be simpler than studying the corresponding PDE. As an example, we analyse the state-dependent distributed DDE in equation~\eqref{Eq:MatureDDE}. Define
\begin{equation*}
\bar{x}(t) = V_a(t) A_K(t) = V_a(t) \int_{0}^{\infty} K_A(a(\phi)) \frac{ \beta x(t-\phi)  }{V_a(t-\phi)}  \exp \left[ - \int_{t-\phi}^t \mu(x(s)) \d s\right] \d \phi,
\end{equation*}
and consider the IVP
\begin{equation}
\left.
\begin{array}{lll}
\TimeDeriv x(t) & = & F\left[ x(t),\bar{x}(t)\right] - \gamma(x(t)) x(t) \quad t > t_0 \\
x(s) & = & \rho(s) \quad  s \in (-\infty,t_0], \\
\end{array}
\right \}
\label{Eq:DDEIVP}
\end{equation}
where $F(x,y) \in\mathcal{C}^1(\mathbb{R}^2,\mathbb{R})$  and $\gamma(x(t)) \in \mathcal{C}^1(\mathbb{R}, \mathbb{R})$ with
\begin{equation}
F(x,y) >0 \quad \textrm{if} \;  x>0 \; \textrm{or} \; y>0, \quad F(0,0)=0, \quad \textrm{and} \quad \gamma(x(t)) <\gamma_{max}<\infty.
\label{Eq:FGammaConditions}
\end{equation}

We recall that $A$ is the random variable representing the maturation age of immature individuals. The history function, $\rho(s)$, is chosen to belong to the space $L_1(A)$ where
\begin{equation*}
K_A(t) = \frac{\d A}{\d \lambda},
\end{equation*}
\sloppy{and $\lambda$ is the Lebesgue measure on $\mathbb{R}$. $L_1(A)$ satisfies the axioms for a phase space given in \citet{Hale1993,Hino1991}, so the solution of the IVP~\eqref{Eq:DDEIVP} exists and is unique in $L_1(A)$.}
In population modelling, it is likely that any realistic history is uniformly continuous and bounded. The space of bounded and uniformly continuous functions is a subspace of $L_1(A)$ and is a suitable phase space.


The age structured PDE~\eqref{Eq:McKendrickAgePDE} describes population dynamics in the presence of a maturation time. Consequently, solutions of \eqref{Eq:DDEIVP} must represent a population, and in particular, remain non-negative. However, the presence of delays in other models may lead to solutions that do not remain non-negative, as noted by \citet{Liu2007}. We begin our analysis by showing that the solution of the IVP~\eqref{Eq:DDEIVP}, $x(t)$, evolving from non-negative initial conditions remains non-negative. This property is a natural requirement for models of population dynamics.

\begin{proposition}\label{Prop:NonNegativitiy}
Let $F(x,y)$ and $ \gamma(x(t))$ satisfy equations~\eqref{Eq:FGammaConditions}
Moreover, assume that the history function satisfies
\begin{equation*}
\rho(s) \geq 0 \quad \forall s \in (-\infty,t_0].
\end{equation*}
Then, the solution of the IVP~\eqref{Eq:DDEIVP} remains non-negative for all time $t> t_0$.
\end{proposition}

\begin{proof}
As $\rho(s) \geq 0$, it is simple to see that
\begin{equation*}
\bar{x}(t_0) = V_a(t_0)\int_{0}^{\infty} K_A(a_{\phi}(t_0)) \frac{ \beta \rho[t_0-\phi] }{V_a(t_0-\phi)}  \exp \left[ - \int_{t_0-\phi}^{t_0} \mu(x(s)) \d s\right]  \d \phi \geq 0.
\end{equation*}
We have a series of cases.

1) If $\rho(t_0) = x(t_0) >0$ then $F(x(t_0),\bar{x}(t_0)) >0$. Therefore,
\begin{equation*}
\TimeDeriv x(t)  =  F(x(t),\bar{x}(t)) - \gamma(x(t)) x(t)\geq -\gamma(x(t)) x(t) > - \gamma_{max} x(t)
\end{equation*}
and using Gronwall's inequality, we have
\begin{equation*}
x(t) \geq \rho(t_0)\exp\left(-\gamma_{max}[t-t_0]\right) > 0.
\end{equation*}

2) If $\rho(t_0) = 0$ and $\rho(s) = 0$ $A$-almost everywhere in $(-\infty,t_0)$, then $x(t)= 0$ is the solution of the IVP.

3) Finally, if $\rho(t_0) = 0$ and $\rho(s)>0$ on a set of $A$-positive measure in $(-\infty,t_0)$ then $\bar{x}(t_0) >0$ and
\begin{equation*}
\TimeDeriv x(t)|_{t=t_0}  =  F(x(t_0),\bar{x}(t_0)) - \gamma(t_0)x(t_0) =   F(0,\bar{x}(t_0)) > 0.
\end{equation*}
Consequently, $x(t)$ becomes positive immediately and case 3 reduces to case 1.

Therefore, solutions of the IVP~\eqref{Eq:DDEIVP} remain non-negative for all time $t>t_0$.
\end{proof}

\subsection{Linearisation of the DDE}

We continue the analysis of equation~\eqref{Eq:MatureDDE} by studying the local stability of equilibrium solutions. To do this, let $x^*(t) = x^* \in L_1(A) $ be an equilibrium of the IVP~\eqref{Eq:DDEIVP}, so
\begin{equation}
F\left( x^*,\bar{x}^* \right) = \gamma(x^*) x^*.
\label{Eq:EquilibriumCondition}
\end{equation}
The homeostatic delayed term $\bar{x}^*$ in \eqref{Eq:EquilibriumCondition} satisfies
\begin{equation*}
\bar{x}^* = \int_0^{\infty}  \beta x^* K_A(V_a^*\phi)\exp\left[ -\mu^*\phi\right] \d \phi = \beta x^* \mathcal{L}[K_A](\mu^*/V_a^*),
\end{equation*}
where $\mathcal{L}[f](s)$ is the Laplace transform of $f(x)$ evaluated at $s$.
Hence, $\bar{x}^*$ is a function of the density $K_A(t)$. However, if desired, it is possible to vary the homeostatic death rate $\mu^*$ to ensure that the equilibria value $x^*$ does not change for different densities $K_A(t)$ as shown by \citet{Cassidy2018}.

Set $z(t) = x(t)-x^*$ and for $z(t)$ small, similar to the discrete state dependent delay case considered by \citet{Hartung2006}, freeze the ageing and clearance rates at their homeostatic rates, so $V_a(t) = V_a^*$  and $\mu(s) = \mu^*$. Then, we define $\bar{z}(t)  =  \bar{x}(t)- \bar{x}^*$ so that
\begin{align}\label{Eq:ZbarDefinition} \notag
\bar{z}(t) &  =  \int_{0}^{\infty} K_A(V_a^*\phi)  \beta x[t-\phi]   \exp \left[ - \mu^* \phi \d s\right] -  \beta x^* K_A(V_a^*\phi)\exp\left[ -\mu^*\phi\right] \d \phi \\
& = {} \int_{0}^{\infty} K_A(V_a^*\phi)  \beta z[t-\phi]   \exp \left[ - \mu^* \phi \d s\right] \d \phi,
\end{align}
to translate the equilibrium to the origin. Then, the differential equation for $z(t)$ is
\begin{equation*}
\TimeDeriv z(t) = F(z(t)+x^*,\bar{z}(t)+\bar{x}^*)-\gamma(x(t))z(t) - \gamma(x(t))x^*.
\end{equation*}
Expanding the exponential integral in \eqref{Eq:ZbarDefinition} following \cite{Cassidy2018}, we find
\begin{equation*}
I := \int_{t-\phi}^t \mu(x(s)) \d s = \int_{t-\phi}^t \mu^* + \mu'(x^*)\left( x(s) - x^* \right) + \mathcal{O}(x(s) - x^*)^2 \d s,
\end{equation*}
so that
\begin{align*}
e^{-I} & = e^{-\mu^* \phi} \left[ 1 - \int_{t-\phi}^t \mu'(x^*)\left( x(s) - x^* \right) + \mathcal{O}(x(s) - x^*)^2 \d s\right] \\
& = e^{-\mu^* \phi} \left[ 1 - \int_{t-\phi}^t \mu'(x^*)z(s) + \mathcal{O}(|z(s)|^2) \d s \right].
\end{align*}
By making the ansatz
\begin{equation*}
z(t) = Ce^{\lambda t},
\end{equation*}
we compute the expression for $\bar{z}(t)$ from \eqref{Eq:ZbarDefinition}
\begin{equation*}
\begin{aligned}
\bar{z}(t) & = Cz(t) \int_{0}^{\infty} K_A(V_a^* \phi) \beta e^{-\lambda \phi} \left[  \exp \left[ - \mu^* \phi \right] + \mathcal{O}(z) \right] \d \phi \\
& =  Cz(t) \beta \mathcal{L}[K_A]([\mu^*+\lambda]/V_a^*) + \mathcal{O}(z^2).
\end{aligned}
\end{equation*}
Therefore, we write
\begin{equation*}
\TimeDeriv z(t) = k_1z(t) + k_2 \beta \mathcal{L}[K_A]([\mu^*+\lambda]/V_a^*) z(t)-\gamma^* z(t) + \mathcal{O}(z^2)
\end{equation*}
where $\gamma^* = \partial_x\gamma(x(t))|_{x=x^*}$, $k_1 = \partial_a F(a,b)|_{(x,\bar{x})}$ and $k_2 = \partial_{b} F(a,b)|_{(x,\bar{x})}$. Dropping nonlinear terms, the linearised equation is
\begin{equation}
\TimeDeriv z(t) = (k_1-\gamma^*)z(t) + k_2 \beta \mathcal{L}[K_A]([\mu^*+\lambda]/V_a^*)z(t) .
\label{Eq:LinearisedEquation}
\end{equation}
The characteristic equation corresponding to \eqref{Eq:LinearisedEquation} is
\begin{equation}
0 = \lambda - (k_1-\gamma^*) - k_2\beta \mathcal{L}[K_A]([\mu^*+\lambda]/V_a^*).
\label{Eq:CharacteristicEquationGeneric}
\end{equation}

Through a standard analysis, we study the local stability of the equilibrium $x^*$ for a density $K_A(t)$. 

\begin{proposition}\label{Prop:StabilityProposition}
1) If
\begin{equation*}
|k_2| \beta \mathcal{L}[K_A](\mu^*/V_a^*)< \gamma^* - k_1,
\end{equation*}
the equilibrium point $x^*$ is locally asymptotically stable.

2) If
\begin{equation*}
k_2 \beta \mathcal{L}[K_A](\mu^*/V_a^*) > \gamma^* - k_1,
\end{equation*}
the equilibrium point $x^*$ is unstable.
\end{proposition}

\begin{proof}
1)  Let $\lambda^*$ be a root of \eqref{Eq:CharacteristicEquationGeneric} and assume for contradiction that $\Re(\lambda^*)  \geq 0. $ We necessarily have
\begin{equation*}
\lambda^* = (k_1-\gamma^*) + k_2 \beta \mathcal{L}[K_A]([\mu^*+\lambda^*]/V_a^*),
\end{equation*}
and we calculate
\begin{equation*}
\Re(\lambda^*) =(k_1-\gamma^*) + k_2 \beta \Re \left[ \mathcal{L}[K_A]([\mu^*+\lambda^*]/V_a^*) \right].
\end{equation*}
We note that
\begin{equation*}
k_2 \beta \Re \left[ \mathcal{L}[K_A]([\mu^*+\lambda^*]/V_a^*) \right] \leq |k_2 \beta \mathcal{L}[K_A]([\lambda^*+\mu^*]/V_a^*)|.
\end{equation*}
While, for arbitrary $\nu = \nu_r+i\nu_i \in \mathbb{C}$,
\begin{equation*}
\begin{aligned}
\left\vert k_2 \beta\mathcal{L}[K_A]([\mu^*+\nu]/V_a^*) \right\vert & = \left\vert k_2\right\vert \beta \left\vert \int_0^{\infty}\exp\left[-(\mu^*+\nu_r+i\nu_i)\phi\right] K_A(V_a^*\phi) \d \phi\right\vert \\[0.2cm]
& \leq   |k_2| \beta\int_0^{\infty} \exp\left[ -(\mu^*+\nu_r)\phi \right] K_A(V_a^* \phi)\left\vert e^{-i\nu_i\phi} \right\vert \d \phi \\[0.2cm]
& =  |k_2| \beta \mathcal{L}[K_A]([\mu^*+\nu_r]/V_a^*).
\end{aligned}
\end{equation*}
Moreover, if $\nu_r  \geq 0$,
\begin{equation*}
|k_2| \beta \mathcal{L}[K_A]([\mu^*+\nu_r]/V_a^*) \leq |k_2| \beta \mathcal{L}[K_A](\mu^*/V_a^*).
\end{equation*}
Therefore, using the assumption in 1), we find
\begin{equation*}
\begin{aligned}
\Re(\lambda^*) & =  (k_1-\gamma^*) + k_2 \beta \Re[ \mathcal{L}[K_A]([\lambda^*+\mu^*]/V_a^*)] \leq (k_1-\gamma^*) + | k_2| \beta \mathcal{L}[K_A](\mu^*/V_a^*) < 0,
\end{aligned}
\end{equation*}
which is a contradiction, so no such $\lambda^*$ can exist.
Therefore, all roots of the characteristic equation have negative real part and the equilibrium is stable.

2) To show instability, we will prove that there must be one characteristic root with positive real part. Define
\begin{equation*}
g(\lambda) := k_1-\gamma-\lambda +  k_2 \beta \mathcal{L}[K_A]([\lambda + \mu^*]/V_a^*),
\end{equation*}
and note that $g$ is continuous with
\begin{equation*}
g(0) = k_1-\gamma +  k_2 \beta \mathcal{L}[K_A]( \mu^*/V_a^*) >0 \quad \textrm{and} \quad \lim \limits_{\lambda \to \infty} g(\lambda) = -\infty.
\end{equation*}
Then, there must be a real $\lambda^* > 0$ such that $g(\lambda^*) =0$. The equilibrium is therefore unstable.
\end{proof}

We note that if $k_2>0$, i.e. the production of mature individuals is controlled through positive feedback with the number of maturing individuals at time $t$, then Proposition~\ref{Prop:StabilityProposition} completely characterizes the local stability of $x^*$. If $k_2<0$, it seems likely that $x^*$ would lose stability through a Hopf bifurcation, similar to the discrete delay case. A similar analysis was done in the constant ageing rate by \citet{Yuan2011}. However, \citet{Yuan2011} did not consider death of immature individuals, nor the linear clearance of mature individuals which corresponds to $\mu = \gamma =0$. 

\section{Distributed Delay Differential Equations with Specific Maturation Probabilities}

Next, we study the DDE found in Theorem~\ref{Theorem:GeneralDDE} for various density functions. By first considering the characteristic equation~\eqref{Eq:CharacteristicEquationGeneric} for specific densities $K_A(t)$, we motivate the reduction of these population models to familiar discrete DDEs and transit compartment ODEs. In the discussion that follows, we once again assume that $x^*\in L_1(A)$ is an equilibrium point so that $\mu(x^*) = \mu^*$ and $V_a(t) = V_a^*$. Denote the homeostatic maturation time as the first moment of the random variable $A$ with constant ageing rate $V_a^*$,
\begin{equation*}
\tau^* = \int_0^{\infty}tK_A(V_a^*t)\d t.
\end{equation*}
Consequently, the expected homeostatic maturation age is given by $\mathcal{T} = V_a^* \tau^*$.

We first consider the degenerate distribution and recover the familiar state dependent discrete DDE. Next, we use a linear chain-type technique to reduce state dependent uniformly distributed DDEs to a system involving two state dependent delays. Finally, we show how to reduce a gamma distributed DDE to a transit compartment system of ODEs.

However, true equivalence between the distributed DDE and the reduced form does not follow directly. We must take care when prescribing initial conditions and history functions so that solutions of the different formulations are in fact equivalent. Only then do these reductions allow for the use of the highly efficient numerical methods available for discrete DDEs and ODEs available in most programming languages.

\subsection{Deterministic Maturation}\label{Sec:StateDependentDDE}

Assuming that maturation is a deterministic process and occurs after achieving the threshold age $\mathcal{T}$ implies that $K_A(t)$ is the degenerate distribution with
\begin{equation}
K_A\left( \int_{t-\phi}^t V_a(s)\d s \right) = \delta \left( \int_{t-\phi}^t V_a(s)\d s-\mathcal{T} \right).
\label{Eq:DeltaDensityDefinition}
\end{equation}
where $\delta(x)$ is the Dirac delta function. In the deterministic case, all individuals mature at precisely the same age $\mathcal{T}$. At the equilibrium $x^*$, using \eqref{Eq:CharacteristicEquationGeneric},  the characteristic equation is
\begin{equation}\label{Eq:DiscreteCharacteristicEquation}
0  = \lambda - (k_1-\gamma^*) - k_2\beta \exp\left[-(\mu^*+\lambda)\mathcal{T}/V_a^* \right]
 =  \lambda - (k_1-\gamma^*) - k_2\beta \exp\left[-(\mu^*+\lambda)\tau^*\right],
\end{equation}
which is exactly the characteristic equation of a discrete DDE. This is unsurprising, since it is well known that threshold conditions lead to discrete DDEs \citep{Otto2017,Smith1993}.

Returning to the DDE~\eqref{Eq:MatureDDE} with $K_A(t)$ given by \eqref{Eq:DeltaDensityDefinition}, the threshold maturation age $\mathcal{T}$ allows us to calculate when an individual that matures at time $t$ began maturation. The maturation time, $\tau(x(t))$, must satisfy the implicit threshold condition
\begin{equation}
\mathcal{T} = \int_{t-\tau(x(t))}^t V_a(s)\d s.
\label{Eq:TauImplicitEquation}
\end{equation}
We use the definition of $\tau(x(t))$ to evaluate the convolution integral given in \eqref{Eq:AKDefinition} to find
\begin{align*}
A_{\delta}(t) & =\int_{0}^{\infty} \delta \left( \int_{t-\phi}^t V_a(s)\d s-\mathcal{T} \right) \frac{ \beta x(t-\phi)  }{V_a(t-\phi)}  \exp \left[ - \int_{t-\phi}^t \mu(x(s)) \d s\right] \d \phi \\
 & = \frac{\beta x[t-\tau(x(t))]}{V_a(t-\tau(x(t)))} \exp \left[ - \int_{t-\tau(x(t))}^t \mu(x(s)) \d s\right].
\end{align*}
Consequently, the corresponding IVP to \eqref{Eq:MatureDDE} with state dependent discrete delay is
\begin{equation}
\left.
\begin{aligned}
\TimeDeriv x(t) & = F\left( x(t), \beta x[t-\tau(x(t))] \exp \left[ - \int_{t-\tau(t)}^t \mu(x(s)) \d s\right]\frac{V_a(t)}{V_a(t-\tau(t))} \right) - \gamma x(t) \\
x(s) & = \rho(s)\quad s \in (-\infty,t_0].
\end{aligned}
\right \}
\label{Eq:StateDependentMatureDDE}
\end{equation}
Choosing the history function for \eqref{Eq:DelayTimeDDE} requires a careful consideration of how $\rho(s)$ controls the ageing velocity $V_a(t)$. For homeostatic histories $\rho(s) = x^*$, we can prescribe $\tau(\rho(s)) = \tau^*$.

To implement \eqref{Eq:StateDependentMatureDDE} numerically, it is necessary to solve \eqref{Eq:TauImplicitEquation} to find the maturation time $\tau(x(t))$. This can be done by
differentiating \eqref{Eq:TauImplicitEquation} to find
\begin{equation}
\TimeDeriv \tau(x(t)) = 1-\frac{V_a(t)}{V_a(t-\tau(x(t)))},
\label{Eq:DelayTimeDDE}
\end{equation}
and imposing the correct initial condition so that the solution of \eq{Eq:DelayTimeDDE} also solves 
\eqref{Eq:TauImplicitEquation}.
In the case that $\rho(s) = x^*$, the it is simple to set $\tau(0)= \tau^*$. However, for more general initial data $\rho(s)$, choosing an appropriate initial condition for \eqref{Eq:DelayTimeDDE} can be delicate \citep{Otto2017}.

Then, we can solve the discrete state dependent DDE by solving the system of equations given by \eqref{Eq:StateDependentMatureDDE} and \eqref{Eq:DelayTimeDDE}.  Hence the age structured PDE framework in Section~\ref{Sec:AgeStructuredPDEReduction} offers an alternative to the ``moving threshold'' method to derive state dependent DDEs as described by \citet{Otto2017}.

\subsection{Uniformly Distributed Maturation}\label{Sec:CompactSupportDistribution}

We consider uniformly distributed DDEs centered about the expected homeostatic maturation age $\mathcal{T}$. In the simplest case, the uniform distribution defines lower and upper threshold ages and assigns equal weight to each age falling between the thresholds. The probability density function corresponding to a uniform distribution centred at $\mathcal{T}$ is
\begin{equation}
K_U(a) = \left\{
\begin{array}{ll}
\frac{1}{2V_a^*\delta} & \textrm{if} \;   a \in [ \mathcal{T}-V_a^* \delta,\mathcal{T}+ V_a^*\delta]\\
0 & \textrm{otherwise}.  \\
\end{array}
\right.
\label{Eq:UniformDensityFunction}
\end{equation}
At the equilibrium $x^*$, with $K_A(t)$ given by the uniform density  \eqref{Eq:UniformDensityFunction}, the characteristic equation \eqref{Eq:CharacteristicEquationGeneric} is
\begin{align}\label{Eq:UniformCharacteristicEquation} \notag
0 & = \lambda - (k_1-\gamma^*) - k_2\beta \frac{1}{2\delta V_a^*[\lambda+\mu^*]/V_a^*}\left[e^{-(\lambda+\mu^*)(\mathcal{T}-V_a^* \delta)/V_a^*}- e^{-(\lambda+\mu^*)(\mathcal{T}+V_a^* \delta)/V_a^*}\right] \\
& {} = \lambda - (k_1-\gamma^*) - k_2\beta \frac{1}{2\delta(\lambda+\mu^*)}\left[e^{-(\lambda+\mu^*)(\tau^*-\delta)}- e^{-(\lambda+\mu^*)(\tau^*+\delta)}\right].
\end{align}
$\mathcal{T}-V_a^*\delta$ and $\mathcal{T}+V_a^*\delta$ represent the minimal and the maximal ages at which an individual can mature. Due to the variable ageing rate, the minimal and maximal delay times, $\tau_{min}(x(t))$ and $\tau_{max}(x(t))$, are state dependent, and implicitly defined by
\begin{equation*}
\mathcal{T}-V_a^*\delta = \int_{t-\tau_{min}(x(t))}^t V_a(s)\d s \quad \textrm{and} \quad \mathcal{T}+V_a^*\delta =  \int_{t-\tau_{max}(x(t))}^t V_a(s)\d s.
\end{equation*}
We note that, at homeostasis, $V_a(s) = V_a^*$ so
\begin{equation*}
\mathcal{T}-V_a^*\delta = \tau_{min}(x^*)V_a^* \quad \textrm{and} \quad \mathcal{T}+V_a^*\delta =  \tau_{max}(x^*)V_a^*.
\end{equation*}
Therefore, the terms $\tau^* - \delta$ and $\tau^*+\delta$ in \eqref{Eq:UniformCharacteristicEquation} correspond to the minimal and maximal homeostatic delay times.

The presence of minimal and maximal delay terms in \eqref{Eq:UniformCharacteristicEquation} hints that a uniformly distributed DDE may be reducible to a discrete DDE with two distinct delays.

Inserting the uniform density \eqref{Eq:UniformDensityFunction} into the convolution integral \eqref{Eq:AKDefinition} gives
\begin{align*}
A_U(t) & = \int_0^{\infty} K_U \left( \int_{t-\phi}^t V_a(s)\d s\right) \frac{ \beta x(t-\phi)  }{V_a (t-\phi)}  \exp \left[ - \int_{t-\phi}^t \mu(x(s)) \d s\right] \d \phi \\
& = \int_{\tau_{min}(t)}^{\tau_{max}(t)} \frac{1}{2\delta} \frac{ \beta x(t-\phi) }{\hat{V}_a (t-\phi)}  \exp \left[ - \int_{t-\phi}^t \mu(x(s)) \d s\right] \d \phi.
\end{align*}
Thus the state dependent uniform distributed DDE is
\begin{equation}
\left.
\begin{aligned}
\TimeDeriv x(t) &= F(x(t),A_U(t)V_a(t)) - \gamma(x(t)) x(t) \\
x(s) &= \rho(s), \quad s \in (-\infty,t_0].
\end{aligned}
\right \}
\label{Eq:UniformDDE}
\end{equation}

\subsubsection{Reduction to Discrete DDE}

Next, we show that \eqref{Eq:UniformDDE} can be reduced to an IVP with two state dependent discrete delays. Once again, this is advantageous, as numerical algorithms for systems of state dependent discrete DDEs are available in most programming languages.

We begin by formalizing the link between uniformly distributed DDEs and discrete DDEs that was hinted at in \eqref{Eq:UniformCharacteristicEquation}. To do this, we proceed similarly to the linear chain technique and show how to write the delayed kernel as the solution of a differential equation. However, unlike the linear chain technique, we will not recover a system of ODEs, but rather a system of differential equations
with two state dependent discrete delays. The technique here can also be adapted to ``tent'' like distributions (see \citep{Teslya2015}).

\begin{lemma}\label{Lemma:UniformDerivative}
$A_U(t)$ satisfies the differential equation
\begin{align}\label{Eq:AUDifferentialEquation}
\TimeDeriv A_U(t)  & =  \frac{1}{2\delta}\left[ \frac{ \beta x[t-\tau_{min}(t)] }{\hat{V}_a (t-\tau_{min}(t))}  \exp \left[ - \int_{t-\tau_{min}(t)}^t \mu(x(s)) \d s\right] \frac{V_a(t)}{V_a(t-\tau_{min}(t))}   \right. \\ \notag
 & \qquad  - \left.  \frac{ \beta x[t-\tau_{max}(t)]  }{\hat{V}_a (t-\tau_{max}(t))}  \exp \left[ - \int_{t-\tau_{max}(t)}^t \mu(x(s)) \d s\right]\frac{V_a(t)}{V_a(t-\tau_{max}(t))}  \right] - \mu(x(t)) A_U(t).
\end{align}
\end{lemma}

\begin{proof}
Similar to the linear chain technique, we differentiate $A_U(t)$ using Leibniz's rule to find
\begin{align*}
\TimeDeriv A_U(t)  & =  \frac{1}{2\delta} \left[  \frac{ \beta x[t-\tau_{max}(t)] }{\hat{V}_a (t-\tau_{max}(t))}  \exp \left[ - \int_{t-\tau_{max}(t)}^t \mu(x(s)) \d s\right]\TimeDeriv \tau_{max}(t) \right. \\
& \quad\qquad {} - \left. \frac{ \beta x[t-\tau_{min}(t)] }{\hat{V}_a (t-\tau_{min}(t))}  \exp \left[ - \int_{t-\tau_{min}(t)}^t \mu(x(s)) \d s\right] \TimeDeriv \tau_{min}(t)\right] \\
& \quad {} + \frac{1}{2\delta} \int_{\tau_{min}(t)}^{\tau_{max}(t)} \TimeDeriv \left( \frac{ \beta x(t-\phi) }{\hat{V}_a (t-\phi)}  \exp \left[ - \int_{t-\phi}^t \mu(x(s)) \d s\right] \right) \d \phi.
\end{align*}
We note that
\begin{align*}
\TimeDeriv \left( \frac{ \beta x(t-\phi)  }{\hat{V}_a (t-\phi)}  \exp \left[ - \int_{t-\phi}^t \mu(x(s)) \d s\right] \d \phi \right)  = & - \frac{\d }{\d \phi} \left( \frac{ \beta x(t-\phi) }{\hat{V}_a (t-\phi)}  \exp \left[ - \int_{t-\phi}^t \mu(x(s)) \d s\right]  \right) \\
& {} \quad  -\mu(x(t))\frac{ \beta x(t-\phi) }{\hat{V}_a (t-\phi)}  \exp \left[ - \int_{t-\phi}^t \mu(x(s)) \d s\right],
\end{align*}
so that, integrating by parts,
\begin{align*}
& \int_{\tau_{min}(t)}^{\tau_{max}(t)} \TimeDeriv \left( \frac{ \beta x(t-\phi) }{\hat{V}_a (t-\phi)}  \exp \left[ - \int_{t-\phi}^t \mu(x(s)) \d s\right] \right) \d \phi  \\
& \qquad\qquad = \left. \left( -\frac{1}{2 \delta} \frac{ \beta x(t-\phi)  }{\hat{V}_a (t-\phi)}  \exp \left[ - \int_{t-\phi}^t \mu(x(s)) \d s\right]\right) \right|_{\phi = \tau_{min}(t)}^{\tau_{max}(t)}  - \mu(x(t))A_U(t).
\end{align*}
Consequently, the derivative of $A_U(t)$ is
\begin{equation*}
\begin{aligned}
\TimeDeriv A_U(t) = & \frac{1}{2\delta}\left[  \frac{ \beta x[t-\tau_{max}(t)]  }{\hat{V}_a (t-\tau_{max}(t))}  \exp \left[ - \int_{t-\tau_{max}(t)}^t \mu(x(s)) \d s\right]\left(\TimeDeriv \tau_{max}(t)-1\right) \right. \\
& \qquad {} - \left. \frac{ \beta x[t-\tau_{min}(t)] }{\hat{V}_a (t-\tau_{min}(t))}  \exp \left[ - \int_{t-\tau_{min}(t)}^t \mu(x(s)) \d s\right] \left( \TimeDeriv \tau_{min}(t)-1\right) \right]-\mu(x(t))A_U(t).
\end{aligned}
\end{equation*}
To finish the proof, we note that, similar to \eqref{Eq:DelayTimeDDE}, $\tau_{min}(x(t))$ and $\tau_{max}(x(t))$ solve the following differential equations
\begin{equation}
\TimeDeriv \tau_{min}(x(t))-1 = -\frac{V_a(t)}{V_a(t-\tau_{min}(x(t)))} \quad \textrm{and} \quad \TimeDeriv \tau_{max}(x(t))-1 = -\frac{V_a(t)}{V_a(t-\tau_{max}(x(t)))}.
\label{Eq:MinMaxDelayDifferentialEquation}
\end{equation}
The identities in equation~\eqref{Eq:MinMaxDelayDifferentialEquation} give \eqref{Eq:AUDifferentialEquation}.
\end{proof}

By writing the delay term $A_U(t)$ as a solution of a differential equation, we are able to reduce the distributed DDE to a system with state dependent discrete delays. Once again, this allows for simulation of the distributed DDE \eqref{Eq:UniformDDE} using existing techniques.  This relationship is formalized in the following theorem.
\begin{theorem}
The IVP \eqref{Eq:UniformDDE} is equivalent to the IVP with the following system of discrete delay differential equations
\begin{equation}
\left.
\begin{aligned}
\TimeDeriv x(t) & = F(x(t),y(t)V_a(t)) - \gamma(x(t)) x(t) \\
\TimeDeriv y(t) & = \frac{1}{2\delta}\left[ \frac{ \beta x[t-\tau_{min}(t)] }{\hat{V}_a (t-\tau_{min}(t))}  \exp \left[ - \int_{t-\tau_{min}(t)}^t \mu(x(s)) \d s\right] \frac{V_a(t)}{V_a(t-\tau_{min}(t))}   \right. \\
 & \quad {} - \left.  \frac{ \beta x[t-\tau_{max}(t)] }{\hat{V}_a (t-\tau_{max}(t))}  \exp \left[ - \int_{t-\tau_{max}(t)}^t \hspace{-0.2cm}\mu(x(s)) \d s\right]\frac{V_a(t)}{V_a(t-\tau_{max}(t))}  \right]\hspace{-0.1cm} - \hspace{-0.1cm}\mu(x(t)) y(t).
\end{aligned}
\right \}
\label{Eq:TwoDelayDiscreteDDE}
\end{equation}
with suitably chosen initial data.
\end{theorem}
\begin{proof}
Using Lemma~\ref{Lemma:UniformDerivative}, it is simple to see that
\begin{equation}
y(t)V_a(t) = A_U(t)V_a(t),
\label{Eq:DiscreteEquivalenceIdentity}
\end{equation}
and the other terms in the differential equations are identical if the initial data are equivalent. It therefore remains to show that we can choose suitable history functions for the distributed and discrete DDEs. For the history function of the distributed DDE \eqref{Eq:UniformDDE}, $\rho(s)$, setting the initial data of \eqref{Eq:TwoDelayDiscreteDDE} to be
\begin{equation*}
x(s) = \rho(s)
\end{equation*}
and
\begin{equation*}
y(t_0) =  \int_{\tau_{min}(t)}^{\tau_{max}(t)} \frac{1}{2\delta} \frac{ \beta \rho(t_0-\phi) }{\hat{V}_a (t_0-\phi)}  \exp \left[ - \int_{t_0-\phi}^{t_0} \mu(\rho(s)) \d s\right] \d \phi.
\end{equation*}
gives the desired equivalence \citep{Teslya2015}. To convert from \eqref{Eq:TwoDelayDiscreteDDE} with history function $x(s) = \eta(s)$ to \eqref{Eq:UniformDDE}, $y(t_0)$ must satisfy
\begin{equation}
y(t_0) =  \int_{\tau_{min}(t)}^{\tau_{max}(t)} \frac{1}{2\delta} \frac{ \beta \eta(t_0-\phi) }{\hat{V}_a (t_0-\phi)}  \exp \left[ - \int_{t_0-\phi}^{t_0} \mu(\eta(s)) \d s\right] \d \phi.
\label{Eq:YICCondition}
\end{equation}
By taking the initial data for \eqref{Eq:UniformDDE} to be $x(s) = \eta(s)$, we see that this condition is sufficient for equivalence of  \eqref{Eq:TwoDelayDiscreteDDE} and \eqref{Eq:UniformDDE}. Now, if \eqref{Eq:YICCondition} does not hold, then \eqref{Eq:DiscreteEquivalenceIdentity} cannot be satisfied at $t =t_0$, so \eqref{Eq:YICCondition} is a necessary and sufficient condition to be able to convert the system of DDEs \eqref{Eq:TwoDelayDiscreteDDE} into the distributed DDE \eqref{Eq:UniformDDE}.
\end{proof}

\subsection{Gamma Distributed Maturation and a Generalized Linear Chain Technique}\label{Sec:GammaDistributedDDE}

Finally, we study gamma distributed DDEs and show how to reduce the state-dependent gamma distributed DDE to a transit chain of ODEs.
The probability density function of the gamma distribution is
\begin{equation}
g_b^j(x) = \frac{b^jx^{j-1}e^{-bx}}{\Gamma(j)},
\label{Eq:GammaDensityFunction}
\end{equation}
where $j,b \in \mathbb{R}$.

Again, let $\mathcal{T}$ denote the mean maturation age and fix $j>0$. Then, we have the following relationships
\begin{equation*}
\mathcal{T} = j/V_a^*, \quad \sigma^2 = j/(V_a^*)^2, \quad \textrm{and} \quad K_{g}(\sigma) = g_{V_a^*}^j(\sigma),
\end{equation*}
where $\sigma^2$ is the variance of the gamma distribution and we set $b=V_a^*$.

Calculating \eqref{Eq:CharacteristicEquationGeneric} for the gamma density in  \eqref{Eq:GammaDensityFunction} gives
\begin{equation}
0= k_1-\gamma - \lambda + k_2 \beta\frac{(V_a^*)^j}{(V_a^*+[\lambda+\mu^*]/V_a^*)^j}.
\label{Eq:GammaDistributedCharacteristicEquation}
\end{equation}
Now, we use the relationships $V_a^* = j/\mathcal{T}$ and $\mathcal{T} = \tau^* V_a^*$ to rewrite the characteristic function as
\begin{align*}
k_1-\gamma - \lambda + k_2 \beta\frac{1}{(1+\frac{\lambda+\mu^*}{V_a^*}^2)^j} & = k_1-\gamma - \lambda + k_2 \beta\frac{1}{(1+\frac{\mathcal{T}(\lambda+\mu^*)}{ V_a^* j})^j} \\
& = k_1-\gamma - \lambda + k_2 \beta\frac{1}{(1+\frac{\tau^*(\lambda+\mu^*)}{j})^j}.
\end{align*}
Using a common denominator gives
\begin{equation}
0 = \left(k_1-\gamma-\lambda\right)\left(1+\frac{\tau^*(\lambda+\mu^*)}{j}\right)^j+k_2\beta
\label{Eq:CommonGammaDistributedCharacteristicEquation}
\end{equation}
Now, we consider multiple cases for the parameter $j$. If $j \in \mathbb{N}$, then \eqref{Eq:CommonGammaDistributedCharacteristicEquation} is a polynominal of degree $j+1$, with $j+1$ roots. This is markedly different than the generic distributed DDE, as the characteristic equation~\eqref{Eq:CharacteristicEquationGeneric} is typically a transcendental functions of $\lambda$ with infinitely many characteristic values. Now, with $j =n/m \in \mathbb{Q}$, we can rearrange \eqref{Eq:CommonGammaDistributedCharacteristicEquation} to
\begin{equation*}
 \left(k_1-\gamma-\lambda\right)\left(1+\frac{\tau^*(\lambda+\mu^*)}{j}\right)^j= -k_2\beta,
\end{equation*}
and raising both sides of the equality to the power $m$ gives
\begin{equation}
0 = (k_1-\gamma-\lambda)^m(1+\frac{\tau^*(\lambda+\mu^*)}{j})^n + \left(-k_2\beta \right)^m .
\label{Eq:RationalGammaCharacterisicEquation}
\end{equation}
Not all solutions of \eqref{Eq:RationalGammaCharacterisicEquation} will necessarily satisfy \eqref{Eq:GammaDistributedCharacteristicEquation}. However, every solution of \eqref{Eq:GammaDistributedCharacteristicEquation} will satisfy \eqref{Eq:RationalGammaCharacterisicEquation}.  Moreover,\eqref{Eq:RationalGammaCharacterisicEquation} is a polynomial with $m+n$ roots, so \eqref{Eq:GammaDistributedCharacteristicEquation} with $j = n/m \in \mathbb{Q}$ has at most $m+n$ roots. However, if the parameter $j$ is not rational, then \eqref{Eq:CommonGammaDistributedCharacteristicEquation} is once again a transcendental equation with possibly infinitely many roots.

The relationship between the number of characteristic values and the parameter $j$ leads to interesting questions. If $j \in \mathbb{N}$ increases by unit steps, then the characteristic equation gains precisely one root. However, if $j$ increases smoothly between $j$ and $j+1$, do characteristic values spring in and out of existence depending on the rationality of $j$? This question, while important, is outside the scope of the current work.

Having studied the characteristic equation of gamma distributed DDEs, we proceed to write down the gamma distributed DDE. We have parametrized the gamma distribution so that at homeostasis, the mean delay time is $\tau^*$. The variable ageing velocity must then be scaled so that at homeostasis, individuals age chronologically. Therefore, we define the scaled ageing velocity
\begin{equation}
\hat{V}_a(t) = \frac{V_a(t)}{V_a^*},
\label{Eq:ScaledAgeingVelocity}
\end{equation}
and will use $\hat{V}_a(t)$ throughout the remainder of our study. The scaled density function $g_{V_a^*}^j(a_t(\phi))$ is given by
\begin{equation*}
g_{V_a^*}^j\left( \int_{t-\phi}^t \hat{V}_a(s)\d s\right) = \frac{(V_a^*)^j}{\Gamma(j)}\left[ \int_{t-\phi}^t \hat{V}_a(s)\d s\right]^{j-1} \exp \left[-V_a^* \int_{t-\phi}^t \hat{V}_a(s)\d s \right].
\end{equation*}
By inserting $g_{V_a^*}^j(a_t(\phi))$ into equation~\eqref{Eq:AKDefinition}, we define
\begin{equation}
 A_{g}(t) =  \int_{0}^{\infty} g_{V_a^*}^j \left( \int_{t-\phi}^t \hat{V}_a(s)\d s\right) \frac{ \beta x(t-\phi) }{\hat{V}_a (t-\phi)}  \exp \left[ - \int_{t-\phi}^t \mu(x(s)) \d s\right] \d \phi.
 \label{Eq:AGammaDefinition}
\end{equation}

Then, the IVP with a state-dependent distributed DDE corresponding to equation~\eqref{Eq:MatureDDE} is
\begin{equation}
\left.
\begin{aligned}
\TimeDeriv x(t)  & =   F \left[ x(t),V_a(t)  A_{g}(t)\right] - \gamma(x(t)) x(t)\\
x(s) & = \rho(s), s \in (-\infty,t_0].
\end{aligned}
\right \}
\label{Eq:GammaDistributedDDE}
\end{equation}
As we show in Section~\ref{Sec:Examples}, equivalent models to \eqref{Eq:GammaDistributedDDE} have been used in pharmacokinetic modelling. However, these models typically take the form of finite dimensional systems of ODEs and the direct link between these ODEs with variable transit rates and \eqref{Eq:GammaDistributedDDE} has not been established previously.

\subsubsection{A Generalized Linear Chain Technique}\label{Sec:FiniteDimensionalRepresentation}

The finitely many roots of equation~\eqref{Eq:GammaDistributedCharacteristicEquation} for integer $j\in \mathbb{N}$ suggest that there is a finite dimensional representation of the DDE~\eqref{Eq:GammaDistributedDDE}. The link between gamma distributed DDEs and transit chain ODEs with constant transit rates has been known since at least \citet{Vogel1961}. The method entered into the English literature in the works of \citet{MacDonald1978} as the linear chain trick or the linear chain technique.

Just as in Section~\ref{Sec:CompactSupportDistribution}, the linear chain technique consists of replacing the convolution integral \eqref{Eq:AGammaDefinition} by the solution of a system of differential equations. To do this, we will exploit the fact that, for $j\in \mathbb{N}$,
\begin{equation}
\frac{\d}{\d x} g_b^1(x) = -b g_b^1(x) \quad \textrm{and} \quad \frac{\d}{\d x} g_b^j(x) = b[g_b^{j-1}(x)-g_b^j(x)].
\label{Eq:GammaDerivatives}
\end{equation}
The linear chain technique has been used extensively in pharmacology to model delayed drug absorption and action. However, typical applications of the technique require that transition rates between compartments are constant and identical. \citet{deSouza2017} developed an adapted linear chain technique that allows for variable transition rates by rescaling time in a non-linear way. This non-linear time rescaling leads to difficulties in establishing a link between time rescaled simulations and time series patient data \citep{deSouza2017}. Here, we provide an alternative technique that allows for variable transition rates between compartments without rescaling time.

We first show how to write \eqref{Eq:AGammaDefinition} as the solution of a system of ordinary differential equations.
\begin{lemma}\label{Lemma:GammaDerivative}
For $j\in \mathbb{N}$, $A_g(t) = x_j(t)$ where $\{ x_i(t) \}_{i=1}^j$ satisfies
\begin{equation}
\begin{aligned}
\TimeDeriv x_1(t) & =  \frac{ \beta x(t) }{\hat{V}_a (t)} - V_a(t)x_1(t) -\mu(x(t)) x_1(t)\\ \notag
\TimeDeriv x_i(t) & = V_a(t)\left[ x_{i-1}(t) - x_{i}(t)\right] - \mu(x(t)) x_i(t)\quad \textrm{for} \quad i = 2,3,...,j.  \notag
\end{aligned}
\label{Eq:TransitChainODEGeneral}
\end{equation}
\end{lemma}
\begin{proof}\phantom{\qedhere}
We first note that
\begin{equation*}
g_{V_a^*}^i\left( \int_{t}^t \hat{V}_a(s)\d s \right) = \left \{
\begin{array}{lll}
V_a^* & \textrm{if} & i =1 \\
0 & \textrm{if} & i  = 2,3,...,j. \\
\end{array}
\right.
\end{equation*}
Then using \eqref{Eq:GammaDerivatives} and \eqref{Eq:CharacteristicLinesEq}, the chain and Leibniz rules show that
\begin{equation*}
\TimeDeriv g_{V_a^*}^1 \left( \int_{\phi}^t \hat{V}_a(s)\d s\right)  =  -V_a^*\hat{V}_a(t) g_{V_a^*}^1 \left( \int_{\phi}^t \hat{V}_a(s)\d s\right) =  -V_a(t) g_{V_a^*}^1\left( \int_{\phi}^t \hat{V}_a(s)\d s\right)
\end{equation*}
while, for $ i = 2,3,4,...$,
\begin{equation*}
\begin{aligned}
\TimeDeriv g_{V_a^*}^i \left( \int_{\phi}^t \hat{V}_a(s)\d s \right)  & =  V_a^* \hat{V}_a(t) \left[ g_{V_a^*}^{i-1}\left( \int_{\phi}^t \hat{V}_a(s)\d s \right) -g_{V_a^*}^i\left( \int_{\phi}^t \hat{V}_a(s)\d s \right) \right]\\
& =  V_a(t) \left[ g_{V_a^*}^{i-1} \left( \int_{\phi}^t \hat{V}_a(s)\d s\right)-g_{V_a^*}^i\left( \int_{\phi }^t \hat{V}_a(s)\d s\right) \right].
\end{aligned}
\end{equation*}
Now  we define,
\begin{equation*}
a(x) = \int_{t-x}^t \hat{V}_a(s) \d s
\end{equation*}
and, for $i=1,2,...,j$,
\begin{equation}
x_i(t) = \int_{-\infty}^{t} g_{V_a^*}^i (a(t-\phi)) \frac{ \beta x(\phi) }{V_a (\phi)}  \exp \left[ - \int_{\phi}^t \mu(x(s)) \d s\right] \d \phi,
\label{Eq:TransitChainDef}
\end{equation}
and note that, after making the change of variable $u = t-\phi$ in $A_g(t)$,
\begin{equation*}
x_j(t) = \int_{-\infty}^{t} g_{V_a^*}^j (a(t-\phi)) \frac{ \beta x[\phi] }{V_a (\phi)}  \exp \left[ - \int_{\phi}^t \mu(x(s)) \d s\right] \d \phi = A_g(t).
\end{equation*}
Now, by differentiating \eqref{Eq:TransitChainDef} using the Leibniz rules, the transit chain $ x_i(t)$ satisfies the following system of equations
\begin{align} \notag
\TimeDeriv x_1(t) & =  \frac{ \beta x(t) }{\hat{V}_a (t)} - V_a(t)x_1(t) -\mu(x(t)) x_1(t)\\ \tag*{\qed}
\TimeDeriv x_i(t) & = V_a(t)\left[ x_{i-1}(t) - x_{i}(t)\right] - \mu(x(t)) x_i(t)\quad \textrm{for} \quad i = 2,3,...,j.
\end{align}
\end{proof}

Importantly, Lemma \ref{Lemma:GammaDerivative} ensures that
\begin{equation}
V_a(t) A_{g}(t) = V_a(t) x_j(t).
\label{Eq:ClosedSystemGammaDistribution}
\end{equation}
Now, we can use the relationship between equations~\eqref{Eq:TransitChainDef} and \eqref{Eq:ClosedSystemGammaDistribution} to establish the following theorem:

\begin{theorem}[Finite Dimensional Representation]\label{Theorem:FiniteDimensionRepresentation}
The distributed state dependent DDE~\eqref{Eq:GammaDistributedDDE} with $j \in \mathbb{N}$ is equivalent to the finite dimensional transit compartment ODE system given by
\begin{equation}
\left.
\begin{aligned}
\TimeDeriv x(t) & = F(x(t),V_a(t)x_j(t) )- \gamma(x(t)) x(t) \\[0.2cm]
\TimeDeriv x_1(t) & =   \frac{ \beta x(t) }{\hat{V}_a (t)} - V_a(t)x_1(t)-\mu(x(t)) x_1(t) \\[0.2cm]
\TimeDeriv x_i(t) & =  V_a(t)\left[ x_{i-1}(t) - x_{i}(t)\right] -\mu(x(t)) x_i(t) \quad \textrm{for} \quad i = 2,3,...,j.
\end{aligned}
\right \}
\label{Eq:GammaEquivalentODE}
\end{equation}
\end{theorem}

\begin{proof}
Lemma~\ref{Lemma:GammaDerivative} ensures that the differential equations are equivalent. Therefore, we need only construct appropriate initial data for the distributed DDE and ODE formulation. For a history function $\rho(s)$ of \eqref{Eq:GammaDistributedDDE}, we set, for $i = 1,2,...,j,$
\begin{equation}
x_i(0) = \int_{-\infty}^{0} g_{V_a^*}^i (a(-\phi)) \frac{ \beta \rho(\phi) }{V_a (\phi)}  \exp \left[ - \int_{\phi}^t \mu(\rho(s)) \d s\right] \d \phi.
\end{equation}
If $\mu(s) = \mu^*$  is constant and the initial conditions satisfy
\begin{equation*}
x_i(0) = \left(\frac{V_a^*}{V_a^*+\mu^*}\right)^i x_1(0),
\end{equation*}
it is simple to choose $\rho(s) = x_1(0)$. However, in the more general case with $\mu(t) \neq \mu^* $ and arbitrary ODE initial conditions $x_i(0)= \alpha_i$ of \eqref{Eq:TransitChainODEGeneral}, we can use a similar method to \citet{Cassidy2018} to construct one of the infinitely many appropriate history functions.
\end{proof}

A form of the expression for the variable age transit chain in equation~\eqref{Eq:TransitChainDef} was derived by \citet{Krzyzanski2011} to study the equivalence between lifespan and transit compartment models in pharmacodynamics. However, the derivation did not include the underlying age structured PDE and was specific to the gamma distribution.  \citet{Gurney1986} derived a similar expression for the density of individuals progressing through a specific stage of maturation from a balance equation. However, they did not explicitly formulate the underlying DDE nor did they derive the correct initial conditions for each of the transit compartments. Consequently, they did not show equivalence between the transit compartment formulation and the DDE.  

\begin{remark}[Recipe for equivalency between ODEs and gamma distributed DDEs]\label{Remark:ODEtoDDERecipe}
We note that the finite dimensional representation of~\eqref{Eq:GammaDistributedDDE} with $j\in \mathbb{N}$ includes a transit compartment chain. Due to the equivalence between \eqref{Eq:GammaDistributedDDE} and \eqref{Eq:GammaEquivalentODE}, we are able to identify the ingredients needed to transform a transit compartment ODE such as \eqref{Eq:GammaEquivalentODE} into a DDE such as \eqref{Eq:GammaDistributedDDE}. We first consider
\begin{equation*}
\TimeDeriv x_1(t)  =   \frac{ \beta x(t) }{\hat{V}_a (t)} - V_a(t)x_1(t)-\mu(x(t)) x_1(t).
\end{equation*}
From the equation for $x_1(t)$, we can easily identify the ratio $\beta x(t)/\hat{V}_a (t)$ as the rate at which individuals in the 1st compartment are created. Next, by considering the rate at which individuals enter the second compartment,
\begin{equation}
\begin{aligned}
\TimeDeriv x_1(t) & =   \frac{ \beta x(t) }{\hat{V}_a (t)} - V_a(t)x_1(t)-\mu(x(t)) x_1(t) \\[0.2cm]
\TimeDeriv x_2(t) & =  V_a(t)x_1(t) - V_a(t)x_{2}(t) -\mu(x(t)) x_2(t),
\end{aligned}
\end{equation}
we find the (possibly variable) transit rate between compartments. Then, a process of elimination immediately yields the mortality rate $\mu(x(t))$ (if $\mu(x(t))<0$, then population growth rather than decay is occurring through the transit chain). The creation and transit rates also yield the homeostatic ageing rate via \eqref{Eq:ScaledAgeingVelocity}. Further inspection of \eqref{Eq:AGammaDefinition} shows that these rates are all that are needed to transform the transit chain ODE to a distributed DDE.
\end{remark}

We note that the classic linear chain technique (see \citet{Smith2011a}) is a special case of Remark~\ref{Remark:ODEtoDDERecipe} where the ageing velocity, $V_a(t)$, is constant.

\section{Examples From Hematopoiesis}\label{Sec:Examples}

Sometimes, analysis of distributed DDEs is more tractable and simpler than that of a high dimensional equivalent ODE system. For example, by rescaling time, \citet{deSouza2017} converted Quartino's  ODE transit compartment model of granulopoiesis into a distributed DDE \citep{Quartino2014}. The distributed DDE formulation proved to be much more analytically tractable than the ODE case, and was used to show the positivity of solutions and establish the local stability of equilibrium solutions.

However, due to the lack of a general numerical algorithm, simulation of distributed DDEs must be handled on a case by case basis. Simulation of transit compartment ODEs is routine in many programming languages and can be used for the calibration of models to existing data. Once calibrated, mathematical models can be simulated and used in a predictive manner. Consequently, by converting models between the equivalent distributed DDE or ODE formulations, researchers can use the form of the model that is most suitable to their needs.

The hematopoietic system controls blood cell production and, through tight cytokine control, is able to quickly respond to challenges, including infection and blood loss. Cytokines control hematopoietic output by varying effective proliferation and maturation rates in each hematopoietic lineage. As cells are not produced instantaneously, there is necessarily a delay between cytokine signal and production response. Mathematical models have been used to understand the complex dynamics observed in so-called dynamical diseases since the 1970s \citep{Glass2015,Rubinow1975,Mackey1978}. Existing mathematical models of hematopoiesis have included discrete, distributed and state-dependent DDEs \citep{Mahaffy1998,Colijn2005a,Crauste2007,Craig2016,Hearn1998a} as well as transit compartment models \citep{Friberg2002,VonSchulthess1982,Krzyzanski2010}.

Here, we use the equivalence between state dependent distributed DDEs and ODE transit compartment models derived in Section~\ref{Sec:FiniteDimensionalRepresentation} to convert two previously published ODE models of hematopoietic cell production to their equivalent state-dependent distributed DDE. The ODE models specify the entrance rate of individuals into the maturation compartment and the maturation speed, $V_a(t)$, which allows for the calculation the birth rate of immature individuals. As these models involve more than one population, the birth rate $\beta$ is no longer constant but is a function of other populations in the model.

In the first example, we show how a model of reticulocyte production can be reduced to a renewal equation whose dynamics are completely characterized by a simple system of ordinary differential equations.

In the second example, we extend the framework of Section~\ref{Sec:FiniteDimensionalRepresentation} to include non-identical transitions between ageing populations and a variable transition rate. This example shows how the state dependent distributed DDE framework addresses the inability of the linear chain technique to model dynamic ageing processes.

\subsection{P\'{e}rez-Ruixo Model of Reticulocyte Production}\label{Sec:PerezModel}

\citet{Perez-Ruixo2008} studied the effect of recombinant human erythropoietin (EPO) on red blood cell precursors using a mathematical model. EPO is the protein responsible for controlling production of red blood cells and their precursors. The model arises from pharmacokinetic and pharmacodynamic data from patients receiving one dose of exogenous EPO.  EPO was modelled through an open two compartment model of exogenous dose absorption and homoeostatic endogenous production rate, $k_{EPO}$, and the blood serum level ($BSL$). The bioavailable exogenous EPO was modelled as a dose dependent hyperbolic function satisfying
\begin{equation*}
F = F_0+ \frac{E_{max}\textrm{Dose}}{ED_{50}+\textrm{Dose}},
\end{equation*}
where $\textrm{Dose}$ is the amount of EPO administered. Exogenous EPO was absorbed through a dual absorption model into the depot and central compartments. The duration of first order absorption into the depot and central compartments are given by $D_1$ and $D_2$, respectively. A fraction of the bioavailable exogenous EPO, $f_r$, was absorbed into the depot compartment before entering the central compartment at rate $k_a$. The depot concentration of EPO follows
\begin{equation}
\TimeDeriv A_1(t) = \left \{
\begin{array}{lll}
\frac{\textrm{Dose} f_r F}{D_1} - k_a A_1 & \textrm{if} & t \leq D_1\\
-k_a A_1 & \textrm{if} & t > D_1,
\end{array}
\right.
\label{Eq:PerezA1Equation}
\end{equation}
The remaining exogenous EPO, $(1-f_r)F$ enters the central compartment following a lag time $t_{\textrm{lag}2}$ and is cleared linearly at the rate $k_{20}$. The volume of the central compartment is $V_1$. The dynamics of exogenous EPO in the central compartment are given by
\begin{equation}
\TimeDeriv A_2(t) = \left \{
\begin{array}{lll}
\frac{\textrm{Dose}(1- f_r) F}{D_2} + k_a A_1(t) + k_{32}A_3(t)- k_{23}A_2(t)  & & \\
\quad  - k_{20}A_2(t) + k_{epo} - \frac{V_{max}A_2(t)/V_2}{K_M+A_2(t)/V_2} & \textrm{if} & t_{\textrm{lag}2}\leq t \leq D_2 \\
k_{epo} - \frac{V_{max}A_2(t)/V_1}{K_M+A_2(t)/V_1} & \textrm{if} & t > D_2, t < t_{\textrm{lag}2},
\end{array}
\right.
\label{Eq:PerezA2Equation}
\end{equation}
Finally, EPO enters the peripheral compartment from -and returns to- the central compartment linearly, so
\begin{equation}
\TimeDeriv A_3(t) = k_{23}A_2(t)-k_{32}A_3(t).
\label{Eq:PerezA3Equation}
\end{equation}
The total bioavailable EPO is given by
\begin{equation*}
C(t) = BSL + A_2(t)/V_1.
\end{equation*}
\citet{Perez-Ruixo2008} considered 4 different pharmacodynamics models of erythrocyte response to exogenous EPO (titled the A,B,C and D models). In each of the 4 different pharmacodynamic models, the EPO dynamics are unchanged and described by equations~\eqref{Eq:PerezA1Equation}, \eqref{Eq:PerezA2Equation} and \eqref{Eq:PerezA3Equation}.

Here, we describe the ``B'' model from \citet{Perez-Ruixo2008}. Model ``B'' divides the erythrocyte progenitors, $P(t)$, into $N_P$ compartments further subdivided into two distinct populations; EPO only affects the growth rate of the first population. Thus, the first $N_P/2$ compartments constitute the EPO sensitive population. Progression through these $N_P$ compartments represents the ageing process of the progenitor cells. Once erythrocyte progenitors have reached maturity, they progress into the reticulocyte population. Once again, the maturation process of reticulocytes is modelled through a series of $N_R$ transit compartments that are not sensitive to EPO. In this manner, the \citet{Perez-Ruixo2008} model uses a concatenation of transit compartments to model the separate ageing processes of  reticulocytes.


The P\'{e}rez-Ruixo ``B'' model of erythrocyte progenitor and reticulocyte production is
\begin{equation}
\left.
\begin{aligned}
\TimeDeriv P_1(t) & = k_{in}- \frac{S_{max}C(t)}{SC_{50}+C(t)}\frac{N_P}{T_P}P_1(t) \\[0.2cm]
\TimeDeriv P_i(t) & = \frac{S_{max}C(t)}{SC_{50}+C(t)} \frac{N_P}{T_P}\left[ P_{i-1}(t)-P_i(t)\right] \quad \textrm{for} \quad i = 2,3,...,N_P/2 \\[0.2cm]
\TimeDeriv P_{N_P/2+1}(t) & = \frac{S_{max}C(t)}{SC_{50}+C(t)}\frac{N_P}{T_P} P_{N_P/2}(t) - \frac{N_P}{T_P} P_{N_P/2+1}(t) \\[0.2cm]
\TimeDeriv P_i(t) & = \frac{N_P}{T_P}\left[ P_{i-1}(t)-P_i(t)\right] \quad \textrm{for} \quad i = N_P/2+2,...,N_P. \\[0.2cm]
\TimeDeriv R_1(t) & = \frac{N_P}{T_P}P_{N_P}(t) - \frac{N_R}{T_R}R_1(t) \\[0.2cm]
\TimeDeriv R_i(t) & =  \frac{N_R}{T_R}\left[ R_{i-1}(t)-R_i(t)\right] \quad \textrm{for} \quad i = 2,3,...N_R.
\end{aligned}
\right \}
\label{Eq:PerezModel}
\end{equation}
By identifying the ingredients necessary from Remark \ref{Remark:ODEtoDDERecipe}, we will show how the distributed DDE framework from Section~\ref{Sec:FiniteDimensionalRepresentation} can account for these separate ageing processes with distinct ageing velocities. Accounting for multiple ageing processes is not possible by rescaling time so approach of \citet{deSouza2017} cannot be generalized to this case.

The most immature erythrocyte progenitors are modelled by $P_1(t)$ and are created from multipotent progenitors differentiating into the erythrocyte lineage at a constant rate $k_{in}$. Transit between the first $N_P/2$ compartments occurs at the variable rate
\begin{equation*}
V_e(t)= \frac{S_{max}C(t)}{SC_{50}+C(t)} \frac{N_P}{T_P} \quad \textrm{with} \quad V_e^* = \frac{S_{max}BSL}{SC_{50}+BSL} \frac{N_P}{T_P}.
\end{equation*}

Using \eqref{Eq:TransitChainDef}, we define $\hat{V}_e(t) = V_e(t)/V_e^*$, so the birth rate of precursor cells into $P_2(t)$ is
\begin{equation*}
 V_e(t) P_1(t) =  \displaystyle \frac{\beta_e(t)}{\hat{V}_{e}(t)}.
\end{equation*}
Further, we see that the only removal of cells from the compartment model is due to transition to later compartments. Therefore, $\mu(t) =0$, and we have identified all the ingredients necessary in  Remark~\ref{Remark:ODEtoDDERecipe}. Therefore, for $ i = 2,3,...N_P/2$,
\begin{equation}
P_i(t) = \int_{-\infty}^t  \frac{V_e(\phi)}{V_e^*} \ P_1(\phi)g_{V_e^*}^{i}\left[\int_{\phi}^t \hat{V}_e(s) \d s\right]\d \phi.
\label{Eq:ErythroProgenitorExpression1}
\end{equation}
The $N_P/2+1$st compartment satisfies
\begin{equation*}
\TimeDeriv P_{N_P/2+1}(t) = V_e(t) P_{N_P/2}(t) - \frac{N_P}{T_P} P_{N_P/2+1}(t).
\end{equation*}
Erythrocyte progenitors enter the first non-EPO sensitive ageing compartment, $P_{N_P/2+1}(t)$, with appearance rate
\begin{equation*}
 \frac{\tilde{\beta}_e(t)}{V_p(t)} = V_e(t)P_{N/2}(t),
\end{equation*}
and then progress through the remaining $N_P/2$ compartments at a constant rate $V_p(t) = V_p^* = N_P/T_P$. Once again, we note that there is no removal of cells in any of the $N_P/2$ compartments, so $\mu(t)= 0$. Further, since the ageing velocity is constant, $\hat{V}_p^* = 1$. Therefore, a simple application of Remark~\ref{Remark:ODEtoDDERecipe} for constant ageing velocity, and using \eqref{Eq:ErythroProgenitorExpression1} gives
\begin{align}\label{Eq:PerezProgenitors}\notag
P_{N_P}(t) & = \int_0^{\infty} \frac{\tilde{\beta}_e(t-\theta)}{N_P/T_P}g_{N_P/T_P}^{N_P/2}(\theta) \d \theta  = \int_{-\infty}^{t} \frac{\tilde{\beta}_e(\theta)}{N_P/T_P}g_{N_P/T_P}^{N_P/2}(t-\theta) \d \theta \\[0.25cm]
&= \int_{-\infty}^{t} \left[\frac{V_e(\theta)}{N_P/T_P} \int_{-\infty}^{\theta}   V_e(\phi) P_1(\phi)g_{V_e^*}^{i-1}\left( \int_{\phi}^{\theta} \hat{V}_e(s) \d s\right)\d \phi \right]g_{N_P/T_P}^{N_P/2}(t-\theta) \d \theta.
\end{align}
Mature erythrocyte precursors enter into the most immature reticulocyte compartment, $R_1(t)$. Given \eqref{Eq:PerezProgenitors}, the differential equation for $R_1(t)$ becomes
\begin{equation*}
\begin{aligned}
\TimeDeriv R_1(t) & = \frac{N_P}{T_P}\underbrace{ \int_{-\infty}^{t} \left[\frac{V_e(\theta)}{N_P/T_P} \int_{-\infty}^{\theta}   V_e(\phi) P_1(\phi)g_{V_e^*}^{i-1}\left( \int_{\phi}^{\theta} \hat{V}_e(s) \d s\right) \d \phi \right]g_{N_P/T_P}^{N_P/2}(t-\theta) \d \theta}_{P_{N_P}(t)} \\
 & \qquad {} - \frac{N_R}{T_R}R_1.
\end{aligned}
\end{equation*}
Hence, the P\'{e}rez-Ruixo ``B'' model of reticulocyte production is equivalent to
\begin{equation*}
\begin{aligned}
C(t) & = BSL + A_2(t)/V_1 \\
\TimeDeriv P_1(t) & = k_{in}- \frac{S_{max}C(t)}{SC_{50}+C(t)}\frac{N_P}{T_P}P_1(t) \\[0.2cm]
\TimeDeriv R_1(t) & =  \frac{N_P}{T_P}\int_{-\infty}^{t} \left[\frac{V_e(\theta)}{N_P/T_P} \int_{-\infty}^{\theta}   V_e(\phi) P_1(\phi)g_{V_e^*}^{N_p/2}\left( \int_{\phi}^{\theta} \hat{V}_e(s) \d s\right) \d \phi \right]g_{N_P/T_P}^{N_P/2}(t-\theta) \d \theta \\
 & \qquad {} - \frac{N_R}{T_R}R_1 \\
 \TimeDeriv R_i(t) & =  \frac{N_R}{T_R}\left[ R_{i-1}(t)-R_i(t)\right] \quad \textrm{for} \quad i = 2,3,...N_R.
\end{aligned}
\end{equation*}
Finally, we can use Remark 4.5 with the constant ageing velocity $V_r(t) = V_r^* =  N_R/T_R$ to solve the transit compartment system for $R_i(t)$ to find
\begin{equation}
R_{i}(t)  = \int_{0}^{\infty} \frac{T_R}{N_R} \beta_R(\sigma)g_{N_R/T_R}^{i}(\sigma) \d \sigma,
\label{Eq:ReticulocyteRenewal}
\end{equation}
where
\begin{equation*}
\beta_R(\sigma) = \frac{N_P}{T_P}\int_{-\infty}^{\sigma} \left[\frac{V_e(\theta)}{N_P/T_P} \int_{-\infty}^{\theta}   V_e(\phi) P_1(\phi)g_{V_e^*}^{N_p/2}\left( \int_{\phi}^{\theta} \hat{V}_e(s) \d s\right) \d \phi \right]g_{N_P/T_P}^{N_P/2}(t-\theta) \d \theta. \\
\end{equation*}
Using the techniques developed in Section~\ref{Sec:FiniteDimensionalRepresentation}, we have transformed the differential equations for the transit compartments for the erythrocyte progenitors and the reticulocytes into renewal type equations given by \eqref{Eq:PerezProgenitors} and \eqref{Eq:ReticulocyteRenewal} \citep{Diekmann2017}. Since \citet{Perez-Ruixo2008} did not model reticulocyte mediated clearance of EPO, the cytokine and early progenitor dynamics are independent of the $P_{N_P}(t)$ and $R_{N_R}(t)$ concentrations. Consequently, the dynamics of equation~\eqref{Eq:PerezModel} are completely determined by the dynamics of
\begin{equation*}
\begin{aligned}
C(t) & = BSL + A_2(t)/V_1 \\
\TimeDeriv P_1(t) & = k_{in}- \frac{S_{max}C(t)}{SC_{50}+C(t)}\frac{N_P}{T_P}P_1(t),
\end{aligned}
\end{equation*}
and the EPO concentrations given by equations~\eqref{Eq:PerezA1Equation}, \eqref{Eq:PerezA2Equation}, and \eqref{Eq:PerezA3Equation}. We are now able to completely characterise the homeostatic behaviour of erythropoiesis by studying
\begin{equation}
\left.
\begin{aligned}
\TimeDeriv A_1(t) & = -k_aA_1(t) \\
\TimeDeriv A_2(t) & = k_{epo} - \frac{V_{max}A_2/V_1}{K_M+A_2/V_1} \\
\TimeDeriv A_3(t) & = k_{23}A_2(t)-k_{32}A_3(t) \\
\TimeDeriv P_1(t) & = k_{in}- \frac{S_{max}C(t)}{SC_{50}+C(t)}\frac{N_P}{T_P}P_1(t),
\end{aligned}
\right \}
\label{Eq:ErythropoiesisRenewalEquation}
\end{equation}
To ensure that the initial value problem \eqref{Eq:ErythropoiesisRenewalEquation} is equivalent to the P\'{e}rez-Ruixo model \citep{Perez-Ruixo2008}, we re-use the initial conditions for $A_1(0),A_2(0),A_3(0)$. Since $\mu=0$ and the initial conditions $P_1(0) = P_i(0) $ are constant, we can set the history function for the progenitors, $\rho_p(s)$, to be $\rho_p(s) = P_1(0)$. The same can be done for the reticulocytes with $\rho_r(s) = R_1(0)$.

We find the homeostatic concentration of EPO in the depot, central and peripheral compartments by solving
\begin{equation*}
\TimeDeriv A_1(t) = 0, \quad \TimeDeriv A_2(t) = 0, \quad \TimeDeriv A_3(t) = 0, \quad \textrm{and} \quad  \TimeDeriv P_1(t) = 0.
\end{equation*}
This yields the following homeostatic EPO concentrations
\begin{equation*}
A_1^* = 0, \quad A_2^* = \frac{V_1k_{epo}k_M}{V_{max}-k_{epo}}, \quad A_3^*= \frac{k_{23}}{k_{32}}A_2^*, \quad \textrm{and} \quad C^* = BSL+A_2^*,
\end{equation*}
while the homeostatic progenitor concentration is
\begin{equation*}
P_1^* = \frac{k_{in}(SC_{50}+C^*)}{S_{max}C^*}\frac{T_P}{N_P}.
\end{equation*}
The simplified erythropoiesis dynamics \eqref{Eq:ErythropoiesisRenewalEquation} and homeostatic concentrations lead to the following proposition:
\begin{proposition}
For positive parameter values, the homeostatic equilibrium point of equation~\eqref{Eq:PerezModel} is locally asymptotically stable.
\end{proposition}
\begin{proof}
The linearisation matrix of equation~\eqref{Eq:ErythropoiesisRenewalEquation} about the equilibrium $x^* = (A_1^*,A_2^*,A_3^*,P_1^*)$ is
\begin{equation*}
\mathbb{J}(x^*) = \left[
\begin{array}{cccc}
-k_a & 0 & 0 & 0 \\
0 & \frac{-V_{max}/V_1 k_M}{(k_M+A_2^*/V_1)^2} & 0 & 0 \\
0 & k_{23} & -k_{32} & 0 \\
0 & \frac{1}{V_1}\frac{S_{max}C^*}{(SC_{50}+C^*)^2} & 0 & -\frac{S_{max}C^*}{SC_{50}+C^*}\frac{N_P}{T_P}\\
\end{array}
\right].
\end{equation*}
The matrix $\mathbb{J}(x^*)$ is lower triangular with strictly negative diagonal entries, so the eigenvalues are strictly negative and the equilibrium is locally asymptotically stable.
\end{proof}

This example illustrates how Remark \ref{Remark:ODEtoDDERecipe} can be adapted to include a series of concatenated ageing processes. In the age structured PDE interpretation, each ageing process corresponds to a unique random variable modelling the transition between distinct stages. As we do not \textit{a priori} expect the transition ages to be independent, interpreting the resulting ageing processes requires some care. The final renewal equation~\eqref{Eq:ErythropoiesisRenewalEquation} includes a joint multivariate distribution representing the concatenation of distinct ageing processes.


Further, \citet{Perez-Ruixo2008} did not show that the homeostatic equilibrium is locally asymptotically stable. For the ODE system \eqref{Eq:PerezModel}, the Jacobian would be a $(3+N_P+N_R) \times (3+N_P+N_R) $ matrix with a degree $(3+N_P+N_R)$ characteristic polynominal. In general, analytically finding the roots of a large degree polynominal is difficult. Hence, while the ODE \eqref{Eq:PerezModel} is obviously finite dimensional, it is analytically intractable.

Conversely, the equivalent renewal equation~\eqref{Eq:ErythropoiesisRenewalEquation} is simple to analyse and a similar argument to Proposition~\ref{Prop:NonNegativitiy} shows that solutions of the renewal equation~\eqref{Eq:ErythropoiesisRenewalEquation} evolving from non-negative initial conditions remain non-negative.  The ``A'', ``C'' and ``D'' models can be modelled as renewal equations through a simple application of the classical linear chain technique and the technique shown here.


\subsection{Roskos's Model of Granulocyte Production}

\citet{Roskos2006} modelled the impact of exogenous administration of granulocyte colony stimulating factor (G-CSF) on neutrophil proliferation and maturation speed. G-CSF is a proinflammatory cytokine that binds to G-CSF specific receptors on mature neutrophil cells and controls neutrophil kinetics through a negative feedback loop \citep{Roberts2005,Shochat2007}. G-CSF governs neutrophil production by increasing the effective proliferation of neutrophil precursors, reducing the maturation time of non-mitotic neutrophil precursors, and increasing release of neutrophil cells from the bone marrow into the blood. The dynamics of neutrophil production have been well-studied from both a mathematical and a pharmacometric point of view \citep{Craig2016,deSouza2017,Quartino2014}. These models have used different techniques to incorporate the delays intrinsic to the system, such as discrete DDEs or transit compartment ODEs.  \citet{Roskos2006} model distinct stages of granulocyte production such as the bone marrow concentrations of metamyelocytes, $M(t)$; band cells, $B(t)$; and segmented neutrophil cells, $S(t)$. The ageing and maturation processes for each of these cell types is modelled through a series of three transit chains with $N_M,N_B$ and $N_S$ compartments, respectively. Moreover, band and segmented neutrophil cells can be shunted into circulation following the administration of G-CSF. We denote the metamyelocyte, band and segmented neutrophil cell shunting rates as $\mu_m(t),\mu_b(t)$ and $\mu_s(t)$

Administration of G-CSF is modelled in a similar way to the EPO model of Section~\ref{Sec:PerezModel} using a first order delayed absorption model. However, \citet{Roskos2006} do not give the differential equations for exogenous administration of G-CSF other than to state that the clearance of G-CSF includes neutrophil receptor mediated clearance through the term
\begin{equation*}
CL_N/F = \frac{k_{cat}/F (B_p(t)+S_p(t))}{K_M+C(t)},
\end{equation*}
where $B_p(t)$ and $S_p(t)$ are the number of circulating band and segmented neutrophil cells, respectively. Due to the feedback between the circulating neutrophil precursors and the cytokine $C(t)$, we are unable to completely reduce the Roskos model to a renewal type equation as was done in Section~\ref{Sec:PerezModel}.

The Roskos model for granulocyte production is
\begin{align*}
\TimeDeriv M_1(t) & = S_0+\frac{E_{mit}C(t)}{EC_{50}+C(t)} - \frac{N_M}{\tau_{meta}\left( 1-\frac{f_{mmt}C(t)}{EC_{50}+C(t)}\right) }M_1(t) \\
\TimeDeriv M_i(t) &  =  \frac{N_M}{\tau_{meta}\left( 1-\frac{f_{mmt}C(t)}{EC_{50}+C(t) } \right) }\left( M_{i-1}(t)-M_{i}(t) \right) \quad \textrm{for} \quad i =2,...,N_M  \\[0.2cm]
\TimeDeriv B_1(t) & = \frac{N_M}{\tau_{meta}\left( 1-\frac{f_{mmt}C(t)}{EC_{50}+C(t)} \right)}M_{N_M}(t) - \frac{N_B}{\tau_{band}\left( 1-\frac{f_{mmt}C(t)}{EC_{50}+C(t)} \right) } B_1(t) \\[0.2cm]
& \qquad {} - \frac{E_{band}C(t)}{EC_{50}+C(t)}B_1(t) \\[0.2cm]
\TimeDeriv B_i(t) & = \frac{N_B}{\tau_{band}\left( 1-\frac{f_{mmt}C(t)}{EC_{50}+C(t)} \right) }\left[ B_{i-1}(t)-B_i(t)\right]  -  \frac{E_{band}C(t)}{EC_{50}+C(t)}B_i(t); \quad i = 2,...N_B \\[0.2cm]
\TimeDeriv B_p(t) & = \displaystyle \sum_{i=1}^{N_B}  \frac{E_{band}C(t)}{EC_{50}+C(t)}B_i(t)- (k_{\lambda}+ k_{bpmat})B_p(t) \\
\TimeDeriv S_1(t) & = \frac{N_B}{\tau_{band}\left( 1-\frac{f_{mmt}C(t)}{EC_{50}+C(t)} \right)}B_{N_B}(t) - \left( \frac{N_S}{\tau_{seg}\left( 1-\frac{f_{mmt}C(t)}{EC_{50}+C(t)} \right) } + \frac{E_{seg}C(t)}{EC_{50}+C(t)}\right) S_1(t)\\[0.2cm]
\TimeDeriv S_i(t) & = \frac{N_S}{\tau_{seg}\left( 1-\frac{f_{mmt}C(t)}{EC_{50}+C(t)} \right) } \left[ S_{i-1}(t) - S_i(t) \right]  - \frac{E_{seg}C(t)}{EC_{50}+C(t)}S_i(t); \quad i = 2,...,N_S. \\
\TimeDeriv S_p(t) & = \displaystyle \sum_{i=1}^{N_S}  \frac{E_{band}C(t)}{EC_{50}+C(t)}S_i(t)- (k_{\lambda}+ k_{bpmat})S_p(t),
\end{align*}
and is an example of a transit compartment model with variable ageing speed and linear clearance. The linear clearance terms are Hill type functions with a maximal clearance rate $E_{j}$ given by
\begin{equation*}
\mu_j(t) =  \frac{E_{j}C(t)}{EC_{50}+C(t)}.
\end{equation*}
Including these linear clearance in a transit compartment model is uncommon, but allows for the direct modelling of G-CSF mediated shunting of immature cells into circulation.

By converting the model into a distributed DDE, we underline the link between clearance of cells in a transit compartment to the exponential decay present in the distributed DDE. Once again, we will proceed by identifying the ingredients discussed in Remark~\ref{Remark:ODEtoDDERecipe}.

As in Section~\ref{Sec:PerezModel}, the most immature metamyelocytes, $M_1(t)$, are produced from the earlier progenitors at a constant baseline rate $S_0$ with the G-CSF dependent recruitment rate
\begin{equation*}
\frac{\beta_m(t)}{V_m(t)} = S_0 + \frac{E_{mit}C(t)}{EC_{50}+C(t)}.
\end{equation*}
Metamyelocytes progress through maturation at a G-CSF dependent rate
\begin{equation*}
V_m(t) = \frac{N_M}{\tau_{meta}\left( 1-\frac{f_{mmt}C(t)}{EC_{50}+C(t)}\right) }.
\end{equation*}
Metamyelocytes are not shunted into circulation following the administration of G-CSF, so $\mu_m(t) =0$. Therefore, the metamylocyte transit compartment model can be reduced to a distributed DDE using Remark~\ref{Remark:ODEtoDDERecipe} in an identical procedure to the P\'{e}rez-Ruixo model in Section~\ref{Sec:PerezModel}. The most mature metamyelocyte population is given by
\begin{equation}
M_{N_M}(t) = \int_{-\infty}^t  \frac{\beta_m(t)}{V_m(t)}g_{V_m^*}^{N_M}\left[\int_{\phi}^t \hat{V}_m(s) \d s\right]\d \phi.
\label{Eq:MetamyelocyteProgenitorExpression}
\end{equation}
Immature neutrophil band cells, $B_1(t)$, are created at the birth rate
\begin{equation*}
\displaystyle \frac{\beta_b(t)}{\hat{V}_b(t)} = \frac{N_M}{\tau_{meta}\left( 1-\frac{f_{mmt}C(t)}{EC_{50}+C(t)} \right)}M_{N_M}(t).
\end{equation*}
These band cells progress through the maturation compartments at the G-CSF dependent ageing rate
\begin{equation*}
V_b(t) = \frac{N_B}{\tau_{band}\left( 1-\frac{f_{mmt}C(t)}{EC_{50}+C(t)} \right) } \quad \textrm{with} \quad  V_b^* = \frac{N_B}{\tau_{band}\left( 1-\frac{f_{mmt}C^*}{EC_{50}+C^*} \right) },
\end{equation*}
so the scaled ageing rate is $\hat{V}_b(t) = V_b(t)/V_b^*$. Inspecting the remaining terms in the equation for $B_1(t)$ gives
\begin{equation*}
\mu_b(t) =  \frac{E_{band}C(t)}{EC_{50}+C(t)}.
\end{equation*}
Therefore, using Remark \ref{Remark:ODEtoDDERecipe}, we find that the $i$-th band compartment satisfies
\begin{equation}
\begin{aligned}
B_{i}(t) & = \hspace{-0.2cm}  \int_{-\infty}^t \frac{\beta_b(\phi)}{V_b(\phi)} \exp\left[-\int_{\phi}^t \mu_b(s) \d s\right]g_{V_B^*}^{i}\left(\int_{\phi}^t \hat{V}_b(s)\d s\right) \d \phi
\end{aligned}
\label{Eq:BandCellExpression}
\end{equation}
 for $i = 1,2,...N_B$.


Mature band cells, given by \eqref{Eq:BandCellExpression} with $i=N_B$, transition into the first segmented neutrophil cell compartment $S_1(t)$ with creation rate
\begin{equation*}
\frac{\beta_s(t)}{\hat{V}_s(t)} = \frac{N_B}{\tau_{band}\left( 1-\frac{f_{mmt}C(t)}{EC_{50}+C(t)} \right)}B_{N_B}(t) = V_b(t)B_{N_B}(t).
\end{equation*}
These cells transit through the segmented neutrophil population with G-CSF dependent ageing ($V_s(t)$) and clearance ($\mu_s(t)$) rates
\begin{equation*}
V_s(t) = \frac{N_S}{\tau_{seg}\left( 1-\frac{f_{mmt}C(t)}{EC_{50}+C(t)} \right) } \quad  \textrm{and} \quad \mu_s(t) = \frac{E_{seg}C(t)}{EC_{50}+C(t)}.
\end{equation*}
Therefore, we have identified all the ingredients in Remark~\ref{Remark:ODEtoDDERecipe} for the segmented neutrophil precursors, $S(t)$. The first segmented neutrophil cell compartment satisfies
\begin{equation*}
\begin{aligned}
\TimeDeriv S_1(t) & = \overbrace{ V_b(t) \int_{-\infty}^t \frac{\beta_b(\phi)}{V_b(\phi)} \exp\left[-\int_{\phi}^t \mu_b(s) \d s\right]g_{V_B^*}^{N_b}\left(\int_{\phi}^t \hat{V}_b(s)\d s\right) \d \phi}^{\beta_s(t)/\hat{V}_s(t)} \\
& \qquad {} -V_s(t) S_1(t) - \mu_s(t) S_1(t).
\end{aligned}
\end{equation*}
Therefore, it is possible to replace the transit compartment system of ODEs for $S_i(t)$ using Remark~\ref{Remark:ODEtoDDERecipe} to find
\begin{align}\label{Eq:SegmentedCompartmentExpression}
\notag
S_i(t) & = \int_{-\infty}^t \overbrace{ V_b(\theta)\left[ \int_{-\infty}^{\theta} \frac{\beta_b(\phi)}{V_b(\phi)} \exp\left[-\int_{\phi}^{\theta} \mu_b(s) \d s\right] g_{V_B^*}^{N_b}\left(\int_{\phi}^{\theta} \hat{V}_b(s)\d s\right) \d \phi\right]}^{\beta_s(\theta)/\hat{V}_s(\theta)}  \\
& \qquad {} \times \exp\left[- \int_{\theta}^{t}  \mu_s(x)\d x \right]g_{V_s^*}^{i}\left(\int_{\theta}^{t} \hat{V}_s(s)\d s\right)  \d \theta \qquad \textrm{for} \quad i = 1,2,...,N_s.
\end{align}

The initial value problem studied by \citet{Roskos2006} was equipped with initial conditions for the cytokine equations as well as the $N_M+N_S+N_B+2$ compartments. Since $\mu \neq 0$ in general, to create an equivalent renewal type equation, we use the same initial conditions as \citet{Roskos2006} for the cytokine differential equations and follow \citet{Cassidy2018} to construct appropriate history functions for $M(t),B(t)$ and $S(t)$.

Therefore, we can reduce the ODE model of granulopoiesis to a renewal-type equation with unchanged cytokine dynamics from \citet{Roskos2006} using the resulting DDEs for $B_p(t)$ and $S_p(t)$. The resulting renewal equation is given by the equations describing the cytokine dynamics and the system of distributed DDEs
\begin{equation*}
\begin{aligned}
\TimeDeriv B_p(t) & = \displaystyle \sum_{i=1}^{N_B}  \frac{E_{band}C(t)}{EC_{50}+C(t)}B_i(t)- (k_{\lambda}+ k_{bpmat})B_p(t) \\
\TimeDeriv S_p(t) & = \displaystyle \sum_{i=1}^{N_S}  \frac{E_{band}C(t)}{EC_{50}+C(t)}S_i(t)- (k_{\lambda}+ k_{bpmat})S_p(t),
\end{aligned}
\end{equation*}
where $B_i(t)$ and $S_i(t)$ are given by \eqref{Eq:BandCellExpression} and \eqref{Eq:SegmentedCompartmentExpression}, respectively.

In this example, we have shown how to concatenate multiple ageing processes with distinct ageing velocities, as well as how to include the loss of cells throughout the ageing process. Once again, we can use a similar argument to Proposition~\ref{Prop:NonNegativitiy} to ensure that the solutions evolving from non-negative initial data remain non-negative.

\section{Discussion}\label{Sec:Discussion}

In this work, we have shown how to reduce age structured PDEs to possibly state-dependent DDEs. Our derivation shows how the correction factor discussed in Section~\ref{Sec:CorrectionFactor} results naturally from considering the hazard rate at which cells exit maturation, and generalises the derivation of \citet{Craig2016} to the non-deterministic case.

In Section~\ref{Sec:AnalysisofDDE}, we analysed the general distributed DDE that arises from the age structured population model. We showed, in Proposition~\ref{Prop:NonNegativitiy}, that populations evolving from non-negative initial conditions remain non-negative, regardless of the density $K_A(t)$. By linearising the distributed DDE, we showed, in Proposition~\ref{Prop:StabilityProposition}, that stability analysis of the general DDE is analytically tractable. We characterized the stability of a generic equilibrium solution as a function of the linearisation of the growth function $F(x^*,\bar{x}^*)$.

Next, we considered the state-dependent DDE in the case of the degenerate, uniform and gamma distribution. Choosing the degenerate distribution leads to the familiar state-dependent discrete DDE, while uniformly distributed DDEs are reducible to discrete DDEs with two state dependent delays. Finally, in the case of gamma distributed DDEs, we explicitly related transit compartment models that include variable transit rates with gamma distributed DDEs in Theorem~\ref{Theorem:FiniteDimensionRepresentation}. As shown by \citet{deSouza2017}, it can be simpler to analyse stability of equilibria and positivity of solutions of a distributed DDE than the corresponding ODE. However, the ODE models may be simpler to simulate numerically. The equivalence between the differential equations allows for the resulting model to be analysed in the more convenient setting.

By the means of two examples, we showed how to express transit compartment models as an equivalent DDE or renewal equation. First, we showed how to incorporate a variable transit rate into a distributed DDE using a simple application of Theorem~\ref{Theorem:FiniteDimensionRepresentation}. Next, we demonstrated that our method is capable of including multiple distinct ageing processes in the form of a multivariate distributed DDE. Lastly, we showed how a linear clearance term in each of the transit compartments can be included in the equivalent DDE model. Analysis of the renewal equation was shown to be simpler than the corresponding ODE system, and we were able to easily characterise the stability of the homeostatic equilibria.

This work emphasizes the link between transit compartment ODEs and delay differential equations. While this link has been known for over 50 years, we explicitly establish it for compartment models with variable transit rates.  We demonstrated that these transit compartment models are equivalent to state dependent distributed DDEs. The equivalence between easy-to-simulate ODE models and the simpler to analyse distributed DDEs allows modellers to use the formulation that is most convenient for their purposes. Consequently, the framework developed in this article allows for researchers to incorporate both external control of ageing rates and heterogeneous, non-deterministic maturation age into models of physiological maturation processes.

\section*{Acknowledgments}
TC would like to thank the Natural Sciences and Engineering Research Council of Canada (NSERC) for funding through the PGS-D program and the Alberta Government for funding through the Sir James Lougheed award of distinction. MC and ARH are grateful for funding through the NSERC Discovery Grant program.


\end{document}